\documentclass{amsart}
\usepackage{epsfig}
\usepackage[matrix,arrow, curve]{xy}
\usepackage{amsfonts}
\usepackage{amsmath}
\usepackage{amssymb}
\usepackage{amsthm}

  \newtheorem{theorem}{Theorem}[section]
  \newtheorem{proposition}[theorem]{Proposition}
  \newtheorem{lemma}[theorem]{Lemma}
  \newtheorem{corollary}[theorem]{Corollary}
  
 \theoremstyle{definition}
 \newtheorem{algorithm}[theorem]{Algorithm}
  \newtheorem{definition}[theorem]{Definition}
  \newtheorem{remark}[theorem]{Remark} 
  \newtheorem{example}[theorem]{Example}
  \newtheorem{examples}[theorem]{Examples}

\newcommand{\abs}[1]{\lvert {#1} \rvert }
\newcommand{\set}[1]{\{#1\}}

\newcommand{\rta}{\rightarrow}
\newcommand{\lta}{\leftarrow}

\newcommand{\Rta}{\longrightarrow}

\newcommand{\rto}{\mapsto}

\newcommand{\eps}{\epsilon}
\newcommand{\lam}{\lambda}
\newcommand{\Lam}{\Lambda}

\newcommand{\sig}{\sigma}

\newcommand{\barf}{\overline{f}}
\newcommand{\barbare}{\overline{\overline{e}}} 
\newcommand{\barr}{\overline{ \;\; }}

\newcommand{\Bzero}{\mathbf{ 0}}
  \newcommand{\Bone}{\mathbf{ 1}}
  \newcommand{\Ba}{\mathbf{ a}}
  
  \newcommand{\Bc}{\mathbf{ c}}
  
  \newcommand{\Bu}{\mathbf{ u}}
  
  \newcommand{\Bdel}{\mathbf{ \delta}}
  
    \newcommand{\BP}{\mathbf{ P}}

  \newcommand{\BT}{\mathbf{ T}}
  \newcommand{\BA}{\mathbf{A}}
  \newcommand{\BK}{\mathbf{K}}
  \newcommand{\BI}{\mathbf{I}}
  \newcommand{\BJ}{\mathbf{J}}
  \newcommand{\BLam}{\mathbf{ \Lam}}
  \newcommand{\BTheta}{\mathbf{ \Theta}}
  \newcommand{\Bystar}{\mathbf{y}^\ast}
    
 \newcommand{\Bz}{\mathbf{ z}}  
  \newcommand{\BZ}{\mathbf{ Z}}
  \newcommand{\BY}{\mathbf{ Y}}
  
\newcommand{\tildeG}{\widetilde{G}}
\newcommand{\F}{{\mathbb F}}
\newcommand{\N}{\mathbb N}

\newcommand{\R}{{\mathbb R}}

\newcommand{\tr}{^\dagger}

\newcommand{\ubm}[2]{\underbrace{#1}_{#2}}

\begin{document}

\title{The Sum-Product Algorithm for Degree-$2$ Check Nodes and
  Trapping Sets
  }
  \thanks{This research was supported by NSF CCF-Theoretical Foundations grants \#0635382 and  \#0635389}

\date{January 4, 2011 version}

\author{John O. Brevik}
\address{Department of Mathematics and Statistics, Long Beach State University}
\email{jbrevik@csulb.edu}

\author{Michael E. O'Sullivan}
\address{Department of Mathematics and Statistics, San Diego State University}
\email{mosullivan@mail.sdsu.edu}

  \maketitle

 \begin{abstract}
 The sum-product algorithm for decoding of binary codes is analyzed for
 bipartite graphs in which the check nodes all have degree $2$.  The
 algorithm simplifies dramatically and may be expressed using linear
algebra.  Exact results about the convergence of the algorithm are
derived and applied to trapping sets.  
 \end{abstract}

\section{Introduction}
\label{s:intro}

One of the great achievements in coding theory in the last decade or
so has been the discovery of iterative decoding methods,
such as the sum-product algorithm (SPA).
Experimental results on very long codes have yielded performance
that is extremely close to Shannon capacity \cite{Mac99},
and asymptotic analysis shows that ensembles of irregular codes
achieve capacity  \cite{RichShokUrb01}.
Aside from the asymptotic theory, which is
presented in detail in \cite{RichUrb:Book},
there is little that has been proven
about the performance of the sum-product algorithm.
The girth and expansion coefficient of the associated bipartite graph,
as well as low-weight pseudo-codewords  and trapping sets (or near-codewords),  are 
thought to affect the performance of the SPA, but we are not aware of
any theorems that quantify the relationship for finite-length codes.
In this article, we focus on a very special case for which we can derive
exact results for convergence of the sum-product algorithm.
By establishing some simple but provable results, we hope to build a
foundation for further algebraic analysis.

We are solely concerned with binary codes. Fix such a code, defined as  
the right null space of a
check matrix $H$.  The SPA is most easily described via the bipartite  
graph of $H$, which has a check node for each row and a bit node for
each column, with an edge between check node $r$ and bit node $\ell$  
if and only if $H_{r\ell}=1$.
The initial data for the algorithm consists of a probability  
distribution for
each bit, indicating the likelihoods that the transmitted signal for  
that
bit is a~$0$ or a~$1$.   The SPA then passes likelihood  data along the
edges of the bipartite graph from the check nodes to the bit nodes and  
back again.
Since we wish to analyze the SPA, we will ignore the binary matrix
defining the code, focusing on the equivalent description via the
bipartite graph.

In this article we study bipartite graphs in which the check nodes all
have  degree~$2$.  The SPA simplifies dramatically and may be
analyzed using  linear  algebra.  The codes defined by such
graphs are repetition codes, provided the graph is connected, and
therefore not of practical utility themselves.
Nevertheless, there are surprising subtleties in the results.
For example, our results indicate that in the simple case of
degree-$2$ checks, 
a covering graph   can inherit provably bad convergence properties from its base graph.
Furthermore, our results can be applied to the dynamics of the SPA in the  
presence of  trapping sets, which have been studied extensively  
 and identified as a powerful influence on decoding  
performance. In essence, we enhance a trapping set by adding some  
nodes that in some sense ``virtually" supply messages from the rest of  
the graph; we then use our methods to study the SPA on this enhanced  
graph.

We are not aware of any previous work on this
class of graphs, but, interestingly, there are
several articles related to the opposite case, in which all bit nodes
have degree $2$.  These codes are called {\em cycle codes} in
\cite{Kot_cycle,Kot_pseudo}, which explore the connection to zeta  
functions of
graphs and the fundamental cone of a parity-check matrix.
Cycle codes were originally called {\em graph theoretic codes}
or {\em circuit codes}  in \cite{HakimiBredeson} and in previous work  
cited
therein.
Cycle codes from Ramanujan graphs are studied in \cite{TillZemor}.

The following section presents some necessary graph-theoretic
terminology and the study of the {\it flow graph} arising from the
SPA.  Section~\ref{s:SPA} presents the sum-product algorithm and
the simplifications due to having all checks of degree~$2$.
Section~\ref{s:convergence} contains the two main theorems on
convergence of the SPA.  Section~\ref{s:examples} presents several
examples to illustrate the results. In Section~6 we analyze the
SPA on trapping sets. In Section~7
we make several observations about the theoretical results
and present some simulations of the SPA on small examples.
These lead to questions for further investigation.

\section{A Discussion of Graphs}
\label{s:graphs}
The sum-product algorithm is defined via a bipartite graph, but our
analysis of the algorithm will use two other graphs derived from the
bipartite graph. 
First, 
a bipartite graph with all check nodes of degree $2$ yields an
undirected graph, and conversely.
Second, the flow of information of the SPA suggests the construction
of a graph whose vertices are the edges of the original bipartite
graph.  This is similar to the traditional notion of a {\em line graph} \cite{Harary:1969}, but with some differences. In particular, our flow graph is a directed graph, containing one vertex for each {\em direction} of each edge in the bipartite graph.
Finally, we also find it useful to allow graphs with loops and with multiple
edges between a given pair of nodes.

In this section, we gather the formal definitions that we will use for
bipartite graph, directed graph, and undirected graph; we present the 
constructions described above; and we establish
some connectivity properties that will be used to analyze the SPA.

\begin{definition}
A {\em directed graph} is a 4-tuple $(E,V, \sig, \tau)$ consisting of
a set $E$ of edges, a set $V$ of vertices 
(or nodes), and two maps $\sig:E\rta V$
and $\tau: E\rta V$ giving the {\em source} and {\em terminus} of an edge.  
\end{definition}

The definition allows for $\sig(e)=\tau(e)$, in which case $e$ is 
a {\it loop}.  The definition also allows for distinct edges $e$ and $f$ to
have both the same source and the same terminus, which gives
{\it parallel edges}.

\begin{definition}
An {\em undirected graph} is a directed graph with an involution
$E \rta E: e \rto  \bar e$ satisfing $\bar e\not= e$, 
$\barbare=e$ and $\tau(\bar e)= \sig (e)$ for all $e\in E$.
\end{definition}

As with a directed graph, loops and parallel edges are allowed.
Note that with this definition, each edge depicted in the conventional drawing of an
undirected graph represents two conjugate edges.  
We will use double-headed arrows to emphasize this aspect of our conventions.

From a given directed graph $G$, there is an obvious way to obtain an
undirected graph $U(G)$: One simply adds a set of conjugate edges,  
$\overline{E}=\{\bar e \mid e\in E\}$.  The undirected graph $U(G)$
has edge set  $E\cup\overline{E}$, vertex set $ V$, the obvious
conjugation map and source and terminus maps extending 
$\sig$ and $\tau$ to $E \cup \overline {E}$ by $\sig(\bar e) =
\tau(e)$ and $\tau( \bar e) = \sig(e)$.

\begin{definition}
A {\it bipartite graph} is a directed graph along with a partition of
the vertex set into two sets $V_1$, $V_2$ such that every edge 
has source in $V_1$ and terminus in $V_2$.
\end{definition}

We will think of the edges going from ``left'' to ``right,'' so we will
write a bipartite graph as a 5-tuple $B=(E,L,R,\lam,\rho)$,  
in which $L$ and $R$ are the source nodes and terminus nodes,
respectively.  The maps  $\lam:E\rta L$ and $\rho:E\rta R$ give the
source  and terminus of an edge. 

For the purpose of error-control coding, the elements of  $L$ are
typically called {\em bit nodes} and the elements of $R$ {\em check nodes}. 
A {\em codeword} is an association of $0$ or $1$ to
each $\ell\in L$ such that each $r\in R$ is connected to an even number
of nonzero bits.  
A binary matrix $H$ yields a bipartite graph by taking $R$ to be the
set of  rows of $H$, $L$ the 
set of columns of $H$ and $E$ enumerating the nonzero entries of $H$, so
that for $e$ the edge associated to
the nonzero entry $H_{r \ell}$, $\lam(e)=\ell$ and $\rho(e)=r$.
A vector in the right nullspace of $H$ gives a codeword for the
associated graph.
Note that while our definition allows for bipartite graphs that have
parallel edges, such a graph clearly does not correspond to a binary
matrix. 

Given an undirected graph $G=(E,V,\sig,\tau)$, we can form a  bipartite graph $(E,L,R,\lambda,\rho)$ in
which all check nodes have degree $2$ by putting a check node on each
edge and making all edges point to the check nodes: Formally, take $L= V$, $\lam = \sig$, and let $R= E/\barr$ be the set $E$ modulo the 
equivalence relation defined by the involution $\overline{\,\,}$; now
let $\rho$ be the $2$-to-$1$ map $\rho: E \rta R$.
The construction is illustrated in
Figure~\ref{f:undirected_bipartite}. 

Conversely, we may form an undirected graph from a bipartite graph 
$B= (E, L, R, \lam, \rho)$ that has  all check nodes of degree~$2$ by
removing each check node and  treating the two edges meeting
at a check node as a conjugate pair.
Formally, for any edge $e\in E$ let $\bar e$ be the unique edge, distinct from
$e$, sharing a node in $R$ with $e$.  Then let $V=L$, $\sig= \lam$ and
let $\tau(e) = \sig(\bar{e})$.  
We will say that this graph is 
{\it the undirected graph associated to} $B$. 
It is clear that the two constructions are inverse operations.
Note that the undirected graph associated to $B$ is only defined
when $B$ has check nodes of degree~$2$, and it is not $U(B)$.

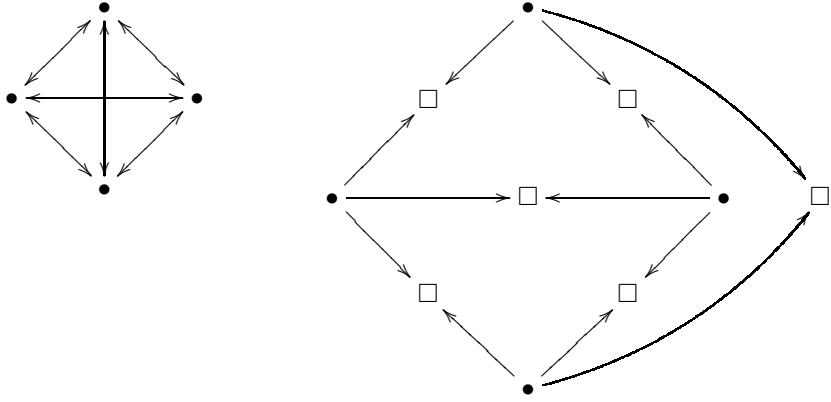
\begin{figure}[hbt]
$$\xymatrix{
& \bullet \ar@{<->}[dr] & \\
\bullet \ar@{<->}[ur]\ar@{<->}[dr] \ar@{<->}[rr] &&\bullet \\
& \bullet \ar@{<->}[ur] \ar@{<->}[uu] & 
}
 \qquad \qquad
\xymatrix{
& & \bullet \ar@{->}[dl]\ar@{->}[dr] \ar@/^1pc/[ddrrr]& &  &\\
&\square && \square &   & \\
\bullet  \ar@{->}[ur]\ar@{->}[dr] \ar@{->}[rr] & & 
\square &&\bullet \ar@{->}[ll]\ar@{->}[ul]\ar@{->}[dl] 
 & \square \\
&\square && \square &   & \\
& & \bullet \ar@{->}[ul] \ar@{->}[ur] \ar@/_1pc/[uurrr] & &  &\\
}
$$

\caption{\label{f:undirected_bipartite}
$K_4$ and the bipartite graph associated to it.
}
\end{figure}

We will see in the next section that the SPA on a bipartite graph with
checks of degree $2$ simplifies to a linear algorithm, which we
describe via the associated undirected graph $G$.  
The flow of information for the SPA follows paths in $G$
that have no backtracking (no edge can follow its conjugate).
Our analysis of the SPA will produce a matrix that is the 
adjacency matrix of a graph $\tildeG$ derived from $G$.
The rest of this section is devoted to the definition and analysis of
$\tildeG$.

Let $G= (V, E, \sig, \tau)$  be a directed graph.  
We write a path in $G$ as a sequence of edges $e_1\circ e_2\circ
\cdots\circ e_n$ such
that $\tau(e_i) = \sig(e_{i+1})$ for $i=1,\dots,n-1$.  
We will say that $\sig(e_1)$  is connected to $\tau(e_n)$ by this
path.  
Notice that the relation ``is connected to'' is transitive, but not
necessarily symmetric.
Recall that $G$ is {\it strongly connected}
when any vertex $v$ is connected to any other vertex $w$.
For an undirected graph  $G$, 
we will say the path  $e_1 \circ e_2 \circ\cdots \circ e_n$ 
is {\it admissible} when  when the path does not involve any
``backtracking," 
that is, $e_{i+1} \ne \bar e_i$ for $i=1,\dots,n-1$.

The path $e_1\circ e_2\circ \cdots\circ e_n$ is a cycle when
$\tau(e_n)=\sig(e_1)$.  
This cycle is {\em completely admissible}
if it is admissible as a path and also $e_n\neq \bar e_1$. Thus a completely
admissible cycle is one that is admissible when traversed starting at
any of its vertices.

Let $G$ be an  undirected graph.  The {\em flow graph} of $G$ is the graph
$\tildeG$ with vertex set $E$ and edge set 
$\set{(e,f) : \tau(e) = \sig(f)\text{ and } \barf \ne e}$.
The source and terminus maps are projections:
$\tilde \sig(e,f) = e$ and $\tilde \tau (e,f) = f$.
There is a natural identification of paths in $\tildeG$ with
admissible paths in $G$.  
An admissible path $e_1\circ e_2 \circ \cdots \circ e_n$ of length $n$
in $G$ yields a path $(e_1,e_2) \circ (e_2,e_3)\circ \cdots \circ 
(e_{n-2}, e_{n-1}) \circ (e_{n-1},e_n)$ of length $n-1$ in $\tildeG$, and conversely.

\begin{remark} The flow graph is different from  the line graph $L(G)$ \cite{Harary:1969},
since we include a vertex in $\tilde{G}$ for each directed edge of
$G$.  It is also different from the usual line graph of a directed
graph \cite{Beineke:1978}, since there is no edge in $\tilde{G}$  between $e$ and
$\bar{e}$.
Stark and Terras \cite[\S3]{StarkTerras}
define an edge zeta function for $G$ by constructing the {\it directed
  edge matrix} of $G$.
Although they do not construct the flow graph, their directed edge matrix
is the adjacency matrix of the flow graph.
\end{remark}

\begin{proposition}
\label{p:cycles}
Let $G$ be an undirected graph, and let $\tildeG$ be as defined above.
There is a natural length-preserving bijection between cycles in
$\tildeG$ and completely admissible cycles in $G$.
\end{proposition}

\begin{proof}
Let $e_1 \circ  e_2 \circ \dots \circ e_n$ be a completely admissible
cycle in $G$. 
Then since $\tau(e_n)= \sig(e_1)$, 
$(e_1,e_2) \circ  (e_2,e_3) \circ \dots \circ (e_{n-2}, e_{n-1})
\circ (e_{n-1},e_n) \circ (e_n,e_1)$ is a cycle in $\tildeG$, also of length $n$.  
Conversely, a cycle in $\tilde G$ is easily seen to give a cycle in
$G$, of the same length, which must be completely admissible.
\end{proof}

The two examples in Figures~\ref{f:triangle}~and~\ref{f:dumbell}
show that $\tildeG$ is not in general an undirected graph.  The examples
also show that $\tildeG$ need not be strongly connected, even when $G$
is.  In Figure~\ref{f:triangle}  the undirected graph
$U(\tildeG)$ derived from $\tildeG$ is not even connected.  
It is clear this phenomenon occurs when $G$ is a path or cycle. 
In Figure~\ref{f:dumbell},  $\tildeG$ is not strongly connected,
since no edge in $\tildeG$ has source~$5$, but $U(\tildeG)$ is
connected.  The following propositions present some connectivity
properties of $\tildeG$.

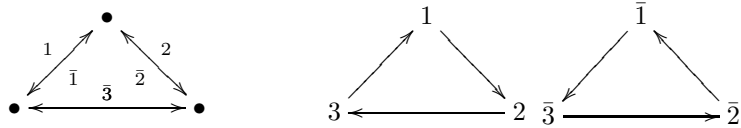
\begin{figure}
$$\xymatrix{
& \bullet \ar@{<->}[dr]^2_{\bar 2}& \\
\bullet \ar@{<->}[ur]^1_{\bar 1} \ar@{<->}[rr]^{\bar 3}^3 & &\bullet 
}
 \qquad \qquad
\xymatrix{
& 1 \ar@{->}[dr]& \\
3 \ar@{->}[ur] \ar@{<-}[rr] & & 2
}
\xymatrix{
& \bar 1 \ar@{<-}[dr]& \\
\bar 3 \ar@{<-}[ur] \ar@{->}[rr] & & \bar 2
}
$$

\caption{\label{f:triangle}
On the left, the graph $G$ is an undirected 3-cycle.  
The outer labels   are for clockwise edges, the inner labels 
  for counterclockwise edges.  On the right, $\tildeG$ consists of two
  disconnected directed $3$-cycles.}
\end{figure}

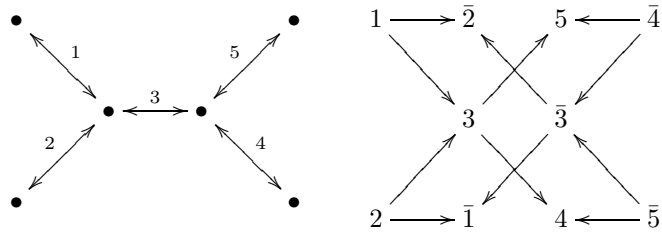
\begin{figure}
$$\xymatrix{
\bullet \ar@{<->}[dr]^1 &&& \bullet \\
& \bullet \ar@{<->}[r]^3 & \bullet \ar@{<->}[dr]^4\ar@{<->}[ur]^5 & \\
\bullet\ar@{<->}[ur]^2 && &\bullet
}
 \qquad 
\xymatrix{
1 \ar@{->}[r] \ar@{->}[dr] & \bar 2 &  5 & \bar 4 \ar@{->}[l] \ar@{->}[dl] \\
 & 3\ar@{->}[ur] \ar@{->}[dr]& \bar 3\ar@{->}[ul]\ar@{->}[dl] & \\
2\ar@{->}[r] \ar@{->}[ur] & \bar 1 & 4 &
\bar 5 \ar@{->}[l] \ar@{->}[ul]
}
$$

\caption{\label{f:dumbell}
On the left, the graph $G$ is an undirected tree.  The edges
1,2,3,4,5 go left to right, and the conjugate edges go right to left.  
On the right the flow graph, $\tildeG$ is not strongly
connected (note, {\em e.g.}, that node $1$ has no edges leading into it), although $U(\tildeG ) $ is connected.}
\end{figure}

\begin{lemma}
\label{eebar} 
Let $G$ be a connected undirected graph. 
Let $e,f$ be edges with $f \ne e$, and $f \ne \bar e$.  
In $\tildeG$, either $e$ or $\bar e$ is connected  either to $f$ or
to $\bar f$ by an admissible path.  
Consequently, $U(\tildeG)$ has at most two connected components.

If $G$ has a vertex of degree at least~3, then $U(\tildeG)$ is connected.
\end{lemma}

\begin{proof}
Let $v=\tau(e), w=\sig(f)$. If $v=w$, then there is an edge $(e,f)$ 
in $\tildeG$ connecting $e$ and $f$, since we assume $f \ne \bar e$.
If $v \ne w$,  then, since $G$ is connected, there is a path
$P$ from $v$ to $w$, and we can eliminate any backtracking to make $P$
admissible.

If $P$ begins with $\bar e$ and ends with $\bar f$, then $\bar e$ is
connected to $\bar f$. 
If $P$ begins with $\bar e$ and does not end with $\bar f$, then
$P\circ \bar f$ is admissible and shows that $\bar e$ is connected to
$\bar f$.

If $P$ does not begin with $\bar e$, then the path $e\circ P$ is still
admissible; now proceed as in the preceding paragraph.

Notice that an admissible path from  $\bar e$ to $f$ 
yields, by reversing direction, an admissible path from $\bar f$ to
$e$.  
We have shown, therefore, that in any conjugate pair $\{f,\bar f\}$ 
either $e$ is connected to one of the edges, or one of the 
edges is connected to $e$.  Consequently,
$U(\tildeG)$ has at most two connected components, one containing the
edges connected by an admissible path to (or from) $e$, and the other those
connected by an admissible path to (or from) $\bar e$.

Suppose $G$ has a vertex $z$ of degree at least~3, and let $e,f,g$ be
three distinct edges with terminus $z$.  
Then $\tildeG$ has the subgraph
$$\xymatrix{
e \ar@{->}[r] \ar@{->}[dr] & \bar f \\
f  \ar@{->}[r] \ar@{->}[dr]&\bar g\\
g \ar@{->}[r] \ar@{->}[uur] & \bar e
}
$$
This shows that $e$ and $\bar e$ are in the same connected component
of $U(\tildeG)$.  By transitivity, $U(\tildeG)$ is connected.
\end{proof}

It is clear that $\tildeG$ is not strongly connected when $G$ has a
vertex of degree one.  The major result of this section,
Proposition~\ref{p:tildeG}, says that
when $G$ is connected, has no vertices of degree $1$, and has one
vertex of degree at least $3$, $\tildeG$ is strongly connected.

\begin{lemma}
\label{admpath} 
Let $G$ be an undirected graph such that
  $G$ is connected, every vertex has degree $\ge 2$, and some vertex
  has degree $\ge 3$. Then  
\begin{enumerate}
\item Every  cycle on $G$ contains a vertex of degree $\ge 3$.
\item Every  edge on $G$ is connected via an admissible
  path, possibly empty,  to an admissible cycle in such a way that neither the path
  nor its conjugate has an edge in common with the cycle.
\end{enumerate}
\end{lemma}

\begin{proof}
Suppose $C$ is a cycle on $G$ all of whose vertices have degree
$2$. $C$ is clearly a connected component of $G$, therefore all of $G$
itself since $G$ is connected; this contradicts the assumption that
$G$ has a vertex of degree $\ge 3$. This establishes the first item in
the theorem.

Let $e$ be a given edge; follow an admissible path from $e$, choosing
edges arbitrarily. Since every vertex of $G$ has degree $\ge 2$, the
path can be extended admissibly at every step until it reaches for the
first time a vertex $w$ already visited. The path from $w$ to $w$ is
necessarily an admissible cycle.  Edges in this cycle, $C$,
share at most one vertex with edges of the path $P$ from $\sig(e)$ to $w$.
Thus there is no edge common to $C$ and either $P$ or $\bar P$.
\end{proof}

\begin{lemma}
\label{l:oppconn}
Let $G$ be as in Lemma~\ref{admpath}. Then every  edge on
$G$ is connected via an admissible path to its opposite. 
\end{lemma}

\begin{proof}
Let $e$ be a given edge. By Lemma~\ref{admpath}, $e$ is connected via
an admissible path $P$ to an admissible cycle $C$.
If $P$ is nonempty, then $P\circ C\circ \bar P$ connects $e$ to $\bar
e$ admissibly. 

If $P$ is empty, then $e$ lies on $C$.  Starting at $v=\sigma(e)$, follow $C$
  along a nonempty path $S$ (possibly equal to $C$ itself) to a vertex
  $w$ of degree $\ge 3$. Then $w$ is the source of an edge $f$ not on
  $C$ or $\bar C$.  By  Lemma~\ref{admpath}, $f$ is connected by a path $Q$ 
 to a simple  cycle $D$. 
\begin{itemize}
\item {Case 1:}   Suppose $Q$ is nonempty, so  $f$ is not part of $D$.
Then $Q\circ D\circ \bar Q $ connects $f$ to $\bar f$ and 
$S\circ Q \circ D \circ  {\bar Q} \circ \bar S$ is an admissible path
from $e$ to $\bar e$.  
\item {Case 2:} 
Suppose $Q$ is empty, so $f$ lies on $D$.  
Let $z$ be the first vertex (beyond the initial $w$)
lying on $D\cap C$, and let $T$ be the corresponding path from $w$ to
$z$ ($z$ may equal $w$). Note that $S \circ T$ is admissible, since the
first edge of $T$ is $f$ which  does not lie on $ C$. 
Now let $U$ be the segment
of $\bar C$ from $z$ back to $v$. Since the final edge of $T$ is not
on $C$, $T\circ U$ is admissible, and therefore so is $S\circ T \circ
U$, connecting $e$ to $\bar e$.  
\end{itemize}
\end{proof}

\begin{examples}
The following graphs illustrate the two cases in the above proof. 
First consider the  undirected graph $G$ derived by adding conjugate
edges to  the following graph:

$$\xymatrix{
\bullet \ar@{->}[dr]^1 & & & \bullet \ar@{->}[dd]^6\\
& \bullet\ar@{->}[r]^4 \ar@{->}[dl]^2 & \bullet\ar@{->}[ur]^5 & \\
\bullet\ar@{->}[uu]^3 & & & \bullet\ar@{->}[ul]^7
}
$$
Note that edge $1$ is connected to edge $\bar 1$ via the edge sequence
$1 \circ 4\circ \ubm{5 \circ 6 \circ 7}{D} \circ \bar 4 \circ \bar 1$.
This is the path constructed in Case~1 of the above proof, 
in which $f$ is edge $4$,  path $S$ is edge~1, and path $T$ is edge~4. 

For Case~$2$,  which $f$ does lie on the cycle, $D$, consider the undirected graph derived from 
\[\xymatrix{
\bullet \ar@{->}[r]^1 \ar@{->}[d]^3  & \bullet \ar@{->}[r]^2 \ar@{->}[d]^4 & \bullet \ar@{->}[d]^5\\
\bullet \ar@{->}[r]^6  & \bullet \ar@{->}[r]^7 & \bullet 
}
\]
To connect edge $3$ to edge $\bar 3$, view edge $3$ as lying on the
cycle $C: 3\circ 6\circ \bar 4 \circ \bar 1$, and follow $\ubm{3\circ
  6}{S}\circ \ubm{7 \circ \bar 5 \circ \bar 2}{T}\circ \ubm{\bar 4
  \circ \bar 6 \circ \bar 3}{U}$. Here of course $f=7$ lies on $D:
7\circ\bar 5\circ \bar 2\circ \bar 1 \circ 3 \circ 6$. Note that it
was important to rejoin the cycle $C$ at the first opportunity, namely at
edge $4$; this is the only juncture that allows us to reverse
directions on $C$. 
\end{examples}

\begin{proposition}
\label{p:tildeG}
Let $G$ be an undirected graph such that
  $G$ is connected, every vertex has degree $\ge 2$, and some vertex
  has degree $\ge 3$.   
Then for any two edges $e,f$ in $G$, $e$ is connected to
$f$ via an admissible path. Consequently, $\tildeG$ is strongly
connected. 

\end{proposition}

\begin{proof}
By Lemma~\ref{eebar}, there is an admissible path $P$ connecting
either $e$ or $\bar e$ either to $f$ or to $\bar f$. If $P$ connects $e$
to $f$, there is nothing to prove.  

If $P$ connects $e$ to $\bar f$, use Lemma~\ref{l:oppconn} to connect $\bar f$ admissibly to $f$ via the admissible path $Q=\bar f\circ Q'$. Now the concatenation $P\circ Q'$ is admissible, since the first edge of $Q'$ cannot be $f$ and so is not the reverse of the last edge $\bar f$ of $P$, and $P\circ Q'$ connects $e$ to $f$. 

If $P$ connects $\bar e$ to $f$, use Lemma~\ref{l:oppconn} to connect $e$ to $\bar e$ by an admissible path $R=R'\circ \bar e$; as in the previous paragraph, $R'\circ P$ is admissible and connects $e$ to $f$.

If $P$ connects $\bar e$ to $\bar f$, use $R'\circ P \circ Q'$ as constructed in the preceding paragraphs.

\end{proof}

\section{The Sum-Product Algorithm with Checks of Degree $2$}
\label{s:SPA}
Throughout this section and the next, let $B=(E,L,R,\lam,\rho)$ be a
bipartite graph.
In this section we show how the sum-product algorithm simplifies when
all check nodes have degree~$2$.

We express all of the probabilistic data  in the sum-product 
algorithm using the {\it odds ratio}, which for a distribution $p$ on
$\{0,1\}$ is $p(1)/p(0)$.  
The  input to the sum-product algorithm  is then
$u_\ell=p_\ell(1)/p_\ell(0)$ where $p_\ell$ expresses the likelihood,
given some received signal for bit $\ell$, that this bit's value is
either $1$ or $0$.
Likewise, the messages along the edges of the graph produced by  the
algorithm are expressed as the odds of $1$.  
 We define the {\it parity} of $u \in (0,\infty)$ to be 0 if $u<1$,
$\infty$ if $u>1$, and undefined when $u=1$.
The sum-product algorithm uses  the transform from the ``odds of $1$''
domain to the ``difference domain'' in which a probability distribution $p$ is
represented using $p(0)-p(1) $, which is in the interval $ [-1,+1]$. 
The function  $s:\R \cup \set{\infty} \Rta \R \cup \set{\infty}$  defined by 
$s(x) = \frac{1-x}{1+x}$ transforms from one domain to the other.
Notice that $s(s(x))=x$.

\begin{algorithm}[Sum-Product Algorithm]\hspace*{1in}

\renewcommand{\descriptionlabel}[1]%
     {\hspace{\labelsep}\textsc{#1}}
\begin{description}
\item [Input:]  For each $\ell \in L$,  $u_{\ell} \in (0,\infty)$. 
Termination criterion $\eps>0$.
\item [Data Structures:] For each $e \in E$, $x_e, y_e \in
(0,\infty)$. 
\item [Initialization:] Set $y_e \lta 1$ for all $e \in E$. \\
\item [Algorithm:] \hspace{1in}
 \begin{description}
\item  [Bit-To-Check Step:] For each $e \in E$, set
$$x_e \lta u_{\lam(e)} \prod_{\substack{f: \lam(f) = \lam(e)\\ f \not=e}} y_f $$
\item [Check-To-Bit Step:] For each $e \in E$, set
$$y_e \lta  s\left( \prod_{\substack{f: \rho(f) = \rho(e)\\f \not= e}}  s(x_f)\right) $$
\item [New Estimate Step:] Set
$$\hat{u}_{\ell} \lta  u_{\ell} \prod_{e \in \lam^{-1}(l)} y_e $$
\end{description}
\item [Termination and Output:] 
If either  $\hat{u}_\ell < \eps$ or  $\hat{u}_\ell > 1/\eps$ for all $\ell
\in L$ then return the binary vector based on the parity of $
\hat{u}_\ell$: Vector $w \in \F^L$ such that   
\[w_\ell = \begin{cases}
1 & \text{ if } \hat{u}_\ell > 1 \\
0 & \text{ else}
\end{cases}
\]
\end{description}
\end{algorithm}

There are a variety of reasonable criteria for termination; we have
simply chosen one.  
We are interested in finding conditions on the set of $u_\ell$ that will determine the convergence behavior of the set of $\hat{u}_\ell$.

When analyzing the algorithm it will sometimes prove useful to
indicate the iteration using a superscript as follows.
We initialize $y^{(0)}=1$ and for $t \geq 1$, 
\[
x_e^{(t)} \lta u_{\lam(e)} \prod_{\substack{f: \lam(f) = \lam(e)\\ f \not=e}}
y_f^{(t-1)}  \qquad \text{and} \qquad
y_e^{(t)} \lta  s\left( 
\prod_{\substack{f: \rho(f) = \rho(e)\\f \not= e}}  s(x_f^{(t)})\right) 
\]

As we have defined it, a bipartite graph is {\em directed};  all edges
go from bit nodes to check nodes.  One could also describe the
algorithm using the undirected graph $U(B)$ discussed in the previous
section. The messages $y_e$ may then be
considered as attached to the reverse arrow $\bar e$.  
The notation proved to be simpler using the directed bipartite graph
since we will define a conjugation map on $E$ itself when all check
nodes have degree~$2$.

We now restrict attention to bipartite graphs in which each check node
has degree~$2$.  We  also assume that the graph is connected,
since the SPA treats each component independently.
One may readily check that the code defined by such a graph is a
repetition code.  The sum-product algorithm simplifies dramatically
because at the check to bit step there is only one term in the
product.

\begin{proposition}
\label{p:monomial}
If all check nodes have degree~$2$, then  then all edge messages are
monomials in the $u_{\ell}$.
Furthermore, for each edge $e$, at any iteration $t$, 
$y_e^{(t)} = x_{\bar e}^{(t)}$ where $\bar e$ is the unique edge sharing a check node
with $e$.
\end{proposition}

\begin{proof}
Clearly, at initialization $y_e^{(0)}=1$ is a monomial, as claimed.
Moreover, if all $y_e^{(t)}$ are monomials, then all $x_e^{(t+1)}$ are monomials as well,
since the bit-to-check step just involves multiplication.
Each right node has degree~$2$, so the product in the check-to-bit
step  has only one term.  Since $s^2$ is the identity, 
$y_e^{(t+1)}=x_{\bar e}^{(t+1)}$ where $e$ is the unique edge distinct from $e$
sharing the same right node. Thus we may establish the proposition
by induction.
\end{proof}

It is now evident that, when all check nodes have degree $2$, the
check nodes play no significant role in the algorithm, and that
analysis of the algorithm is simplified by keeping track only of the
exponents for the monomials $x_e^{(t)}$.
We can express the SPA using the  undirected graph $G$  associated to
$(E,L,R,\lam,\rho)$ ({\em cf.} \S 2).
Let us use $\Ba_e \in \N^L$ to denote the row vector of exponents
appearing in $x_e$, so  
$x_e = \prod_{\ell \in L} u_{\ell}^{a_{e,\ell}}$. 
We will abbreviate this product as  $ \Bu^{\Ba_e}$.
When we want to specify the $t^{\rm th}$ iteration we will write $\Ba_e^{(t)}$.
Let $\Bzero \in \N^L$ be the all-$0$ row vector and let $\Bdel_\ell \in
\N^L$ be the  
row vector which is $1$ in the $\ell^{\rm th}$ component and 0 otherwise.
The update in the SPA is 
$x_e \lta u_{\lam(e)} \prod_{\substack{f: \lam(f) = \lam(e)\\
    {f}\not=e}} x_{\bar f}$, or in terms of $\Bu$ and $\Ba_e$, 
\[
x_e = \Bu^{\Ba_e} \lta u_{\lam(e)} 
\prod_{\substack{f: \lam(\bar f) = \lam(e)\\ {\bar f }\not=e}} \Bu^{\Ba_f}
\]
Keeping track of the exponents gives us the following algorithm.

\begin{algorithm}[Local Sum Algorithm]\hspace*{1in}

\renewcommand{\descriptionlabel}[1]%
     {\hspace{\labelsep}\textsc{#1}}

\begin{description}
\item [Data Structures:] For each $e \in E$, $\Ba_e,  \in \N^L$. 
\item [Initialization:] Set $\Ba_e \lta \Bzero $ for all $e \in E$. \\
\item [Algorithm:] Set \hspace{1in} 
\begin{align}
\Ba_e &\lta \Bdel_{\lam(e)} + \sum_{\substack{f: \lam(\barf) =
    \lam(e)\\ \barf \not=e}} \Ba_{f}  
\label{e:local sum}
\end{align}
\end{description}
\end{algorithm}

Let $\BA$ be the $\abs{E} \times \abs{L}$ matrix whose $e^{\rm th}$ row is
$\Ba_e$.  Let  $\BLam$ be the $\abs{E} \times \abs{L}$ matrix and let
$\BK$ be the $\abs{E}\times \abs{E}$ matrix defined by
\[
\BLam_{e,l} = 
\begin{cases}
1 & \text{  when } \lam(e)=l \\
 0& \text{ else}
\end{cases}
\qquad 
\BK_{e,f} = \begin{cases}
1 & \text{ when }\lam(\barf) = \lam(e) \text{  and }\barf \ne e \\
 0& \text{ else}
\end{cases}
\]
The $e^{\rm th}$ row of $\BLam$ is $\Bdel_{\lam(e)}$ and 
one can check that  $\BK$ is the adjacency matrix of the graph
$\tildeG$.

The local sum algorithm is then 
\begin{align*}
 \BA^{(0)} &= \Bzero \\
\BA^{(t)} &= \BLam + \BK \BA^{(t-1)}
\end{align*}
The equation above is easily solved, for $t\geq 1$, 
\[ \BA^{(t)} =  \Big(\BK^{t-1} + \BK^{t-2} +  \dots + \BK + \BI \Big) \BLam
\]

\section{Convergence of the SPA}
\label{s:convergence}

Throughout this section we adopt the following notation.
Let $B=(E, L, R, \lam, \rho)$   be a bipartite graph,
We assume that $U(B)$ is connected and that 
\begin{itemize}
\item all check nodes have degree 2,
\item all bit nodes have degree at least 2,
\item some bit node has degree at least 3.
\end{itemize}
Let $G$ be the associated undirected graph as constructed in
Section~\ref{s:graphs} and let $\tildeG$ be the flow graph of $G$.  
Let $\BK$ be the adjacency matrix of $\tildeG$, which is the matrix
of the local sum algorithm.  

A vector or matrix is said to be {\em nonnegative}
if each of its entries is nonnegative, so $K$ is nonnegative.
Similarly, a vector or matrix is {\em positive} when each entry is
positive. 
We first treat the case in which $\BK$ is {\em primitive}, that is,
$\BK^r$ is positive for some positive integer $r$.
By the Perron-Frobenius theorem \cite[8.2,5]{Horn},
\cite[Ch. 1]{Minc} 
$\BK$ has a real positive eigenvalue $\rho$ satisfying the following.
\begin{itemize}
\item $\rho$ is  between the maximum row sum and the minimum
  row sum of   $\BK$, and strictly  between the two if the maximum and minimum
  are not equal.  
\item $\rho$ has algebraic multiplicity $1$.
\item $\rho$ is strictly larger in modulus than all other eigenvalues
  of $\BK$.
\item The eigenvectors  associated to $\rho$ are strictly positive.
\item Let $\Bz$ be a right eigenvector and $\Bystar$  a left eigenvector
  associated to $\rho$, normalized so that $\Bystar \Bz=1$. 
Then \cite[8.2.7]{Horn} $\BK^i = \rho^i \Bz\Bystar +(\BK-\rho \Bz\Bystar)^i$
and the eigenvalues of $(\BK-\rho \Bz\Bystar)^i$ have modulus less than
$\rho$.
\end{itemize} 
The eigenvectors $\Bystar$ and $\Bz$ are called {\em Perron vectors} of
$\BK$. 
As a consequence of the final point, for large powers of $i$,
$\BK^i$ is approximated by $\rho^i \Bz\Bystar$.

\begin{theorem}
\label{t:primitive}
With the notation above, let $\Bc = \Bystar \BLam$.  The sum-product
algorithm on $B$  converges based on the parity of $\Bu^{\Bc}$.
That is, the algorithm converges to $0$ when $\Bu^{\Bc} <1$ and to
$\infty$ when $\Bu^{\Bc} >1$.
\end{theorem}

\begin{proof}
Since we assumed that the bipartite graph $B$ has all vertices of
degree at least~$2$ and one vertex of larger degree, 
the row sums of $\BK$ are at least~1 and strictly greater than $1$ for
some edge.  Thus $\rho>1$.
We have $\BK^i = \rho^i \Bz\Bystar + (\BK-\rho \Bz\Bystar)^i$ and therefore
\begin{align*}
\sum_{i=0}^{t-1} \BK^i &= \dfrac{\rho^t -1}{\rho-1} \Bz\Bystar + 
\sum_{i=0}^{t-1}(\BK-\rho \Bz\Bystar)^i \\
\dfrac{\rho -1}{\rho^t-1} \sum_{i=0}^{t-1} \BK^i &=  \Bz\Bystar +
\dfrac{\rho -1}{\rho^t-1}\sum_{i=0}^{t-1}(\BK-\rho \Bz\Bystar)^i  
\end{align*}
Since the eigenvalues of $(\BK-\rho \Bz\Bystar)  $ are less than $\rho$ in
modulus, the final term in the last equation goes to $0$ as $t $ goes
to $\infty$.  Thus  
\begin{align*}
\lim_{t \rta \infty} \dfrac{\rho -1}{\rho^t-1} \BA^{(t)}
& = \Bz\Bystar \BLam\\
\lim_{t \rta \infty} \dfrac{\rho -1}{\rho^t-1} \Ba_e^{(t)} &= \Bz_e \Bystar \BLam \\
\lim_{t \rta \infty}\left( x_e ^{(t)}\right) ^{\frac{\rho -1}{\rho^t-1}}
&= 
\lim_{t \rta \infty} 
(\Bu^{\Ba_e^{(t)} })^{\frac{\rho -1}{\rho^t-1}} \\ 
&= (\Bu^{\Bc})^{\Bz_e }
\end{align*}
where $\Bc = \Bystar \BLam$.
Since $\rho>1$,  $0<\frac{\rho-1}{\rho^t-1}<1$.  
Thus for any edge $e$, $x_e$ goes to 0 if $\Bu^{\Bc}<1$ and to 
$\infty$ if $\Bu^{\Bc} >1$. 
\end{proof}

The proof gives information about the rate of convergence, which,
for all edges, is roughly exponential with  exponent $\rho$,
$x_e^{(t)} \approx \left( \left(\Bu^{\Bc}\right)^{\Bz_e}\right)^{ \frac{\rho^t -1}{\rho-1}}$.
The base for the exponential growth, $(\Bu^{\Bc})^{\Bz_e }$, depends upon the edge.
The following corollary summarizes this result, and the analagous 
statement for the rate of convergence of the
new estimates, $\hat u_\ell$, which varies with $\ell$. 

\begin{corollary}
\label{c:primitive}
The limiting  behavior of $x_e$ and $\hat{u}_\ell$ at iteration $t$ are
as follows.
\begin{align*}
\lim_{t \rta \infty}\left( x_e^{(t)}\right) ^{\frac{\rho -1}{\rho^t-1}}
&= (\Bu^{\Bc})^{\Bz_e } \\
\lim_{t \rta \infty}\left( \hat{u}_\ell ^{(t)} \right) ^{\frac{\rho -1}{\rho^t-1}}
&= (\Bu^{\Bc})^{\sum_{e \in \lam^{-1}(\ell)} \Bz_e }
\end{align*}
\end{corollary}

Let us now turn to the general case, in which $\BK$ may not be
primitive.  Our assumptions on the bipartite
graph $B$ ensure that $\tildeG$ is  strongly connected.  
This is equivalent to $\BK$ being {\em irreducible} 
\cite[Thm. 6.2.24]{Horn} \cite[Ch. 4, Thm. 3.2]{Minc}, since it is the
adjacency matrix of $\tilde G$.  
Thus we are led to the theory of irreducible matrices.

A  square matrix $H$ is {\it  reducible} when there exists a
permutation matrix $P$ such that $P\tr HP= \begin{bmatrix}
A & B \\
0 & C
\end{bmatrix}$ with $A$ and $C$ square matrices.
When $H$ is not reducible, it is called irreducible.
If $H$ is the adjacency matrix of a graph, then it is straightforward
to show that $H$ is reducible if and only if  there is a nontrivial
partition of the vertex set $V$ into $V_1$, $V_2$ such that there is
no edge from $V_2$ to $V_1$. 
There are several other equivalent conditions for irreducibility in 
\cite[\S6.2]{Horn} and in \cite[\S1.2]{Minc}.  

The important properties for analysis of the sum-product algorithm
appear in \cite[Ch.3, Ch. 4 \S3]{Minc}.
Since $\BK$ is irreducible, it has a positive eigenvalue $\rho$ of
maximum modulus.   
There is a positive integer $h$,
called the {\it index of imprimitivity} of $\BK$, satisfying the following
equivalent conditions.
\begin{itemize}
\item $\BK$ has $h$ eigenvalues of modulus $\rho$.  
\item $h$ is the largest positive integer such that $\BK^h$ 
is a block-diagonal matrix with $h$  blocks, each an irreducible matrix.
\item $h$ is the largest integer for which there is a partition of $E$
  into disjoint sets $E_1,\dots, E_h$ such that any edge of $\tildeG$  
  goes from $E_{i+1}$ to $E_{i}$ or $E_1$ to $E_h$.
\item The greatest common divisior of the lengths of the cycles in
  $\tildeG$ is $h$.
\end{itemize}
Much more is known.
The $h$ eigenvalues of $\BK$ with modulus $\rho$ are  $\rho\zeta^i$ where
$\zeta$ is an  $h^{\rm th}$ root of unity. 
Enumerating the elements of $E$ by first taking  
the elements of $E_1$, then $E_2$, and continuing on to $E_h$, the
matrix $\BK$ has the form 
 \[
\BK = \begin{bmatrix}
\Bzero & \BK_1 & \Bzero & \Bzero & \dots & \Bzero\\
\Bzero & \Bzero & \BK_2 & \Bzero & \dots &\Bzero\\
\Bzero & \Bzero & \Bzero & \BK_3 & \dots &\Bzero\\
\Bzero & \Bzero & \Bzero & \Bzero & \dots & \BK_{h-1}\\
\BK_h & \Bzero & \Bzero & \Bzero & \dots &\Bzero\\
\end{bmatrix}
\]
with square $\Bzero$ matrices along the diagonal.  
Computing $\BK^h$, one can see  that 
it is block diagonal with $j^{\rm th}$ block equal to 
$\BK_j\BK_{j+1} \cdots \BK_{h}\BK_1\cdots \BK_{j-1}$.
The conditions above state that this product is irreducible for all
$i$, but the maximality of $h$ actually ensures that these products
are all primitive. 
Furthermore, they have the same nonzero eigenvalues, since for any
matrices $A$, $B$ with compatible dimensions the nonzero eigenvalues
of $AB$ and $BA$ are the same.
Finally, we note that Proposition~\ref{p:cycles}  shows that the gcd
of the cycle lengths in $\tilde G$ equals the gcd of the lengths of
completely admissible cycles in $G$. 

Let $\Bystar_1$ and $\Bz_1$ be  Perron  vectors for the
product  $\BK_1\BK_2 \cdots \BK_h$
satisfying $\Bystar_1\Bz_1=1$.  For $j=1,\dots,h$, define 
\begin{align*}
\Bystar_j &= \rho^{1-j}\Bystar_1 \BK_1\BK_2\cdots \BK_{j-1}\\
\Bz_{j} &= \rho^{j-h-1} \BK_j\BK_{j+1}\cdots \BK_h \Bz_1
\intertext{Now, $\Bystar_j$ and $\Bz_j$
are Perron vectors for $\BK_j\BK_{j+1} \cdots \BK_{h}\BK_1\cdots
\BK_{j-1}$, as is readily verified; moreover,}
\Bystar_j \Bz_j &= \rho^{-h} \Bystar_1 \BK_1\BK_2\cdots \BK_h \Bz_1 = 1 \\
\Bystar_j \BK_j &= \rho \Bystar_{j+1}  \\
\BK_j \Bz_{j+1} &= \rho \Bz_j 
\end{align*}
where the subscripts are computed modulo $h$, with representatives
$\set{1,2,\dots,h}$. 

Note that $\Bz_j$ is a column vector with $\abs{E_j}$ entries and
$\Bystar_j$ is a row vector with $\abs{E_j}$ entries.
We form the $\abs{E} \times h$ matrix $\BZ$ and the $h \times \abs{E}$
matrix $\BY$ as follows:
\begin{align*}
\BZ = \begin{bmatrix}
\Bz_1 & \Bzero & \Bzero &\dots & \Bzero \\
\Bzero & \Bz_2 & \Bzero  &\dots & \Bzero \\
\Bzero & \Bzero & \Bz_3  &\dots & \Bzero \\
\dots &  \dots &  \dots &  \dots &\dots   \\
\Bzero & \Bzero & \Bzero & \dots &\Bz_h 
\end{bmatrix},
&\qquad
\BY = \begin{bmatrix}
\Bystar_1 & \Bzero & \Bzero &\dots & \Bzero \\
\Bzero & \Bystar_2 & \Bzero  &\dots & \Bzero \\
\Bzero & \Bzero & \Bystar_3  &\dots & \Bzero \\
\dots &  \dots &  \dots &  \dots &  \dots \\
\Bzero & \Bzero & \Bzero & \dots &\Bystar_h 
\end{bmatrix}
\intertext{We also form the $h\times h$ matrices}
\BTheta 
= \begin{bmatrix}
1 & \rho & \rho^2 &\dots & \rho^{h-1} \\
\rho^{h-1} & 1 & \rho  &\dots & \rho^{h-2} \\
\rho^{h-2} & \rho^{h-1} &1  &\dots & \rho^{h-3} \\
\vdots &  \vdots &  \vdots &  \vdots &  \vdots \\
\rho & \rho^{2}  & \rho^3 & \dots & 1  
\end{bmatrix},
&\qquad \BP= \begin{bmatrix}
0 & 1 & 0 &0 &  \dots & 0 \\
0 & 0 & 1 & 0 & \dots & 0 \\
0 & 0 & 0 & 1 & \dots & 0 \\
\vdots &  \vdots & \vdots &  \vdots &  \vdots &  \vdots \\
1 & 0 & 0 &0 &  \dots & 0 
\end{bmatrix}
\end{align*}

\begin{lemma}
\label{l:not primitive}
With the matrices $\BZ, \BY, \BTheta,\BP$ defined above and for any
 $0\leq r <h$,
\[
\lim_{q \rta \infty} \dfrac{\rho^h-1}{\rho^{hq}-1}
\sum_{i=0}^{hq+r-1} \BK^i 
= 
\rho^r \BZ\BP^r \BTheta  \BY
\]
\end{lemma}

\begin{proof} 
First we write
\begin{align*}
\sum_{i=0}^{hq+r-1} \BK^i &= 
\BK^r\Big(\sum_{i=0}^{hq-1} \BK^i \Big)  + \Big( \BK^{r-1} + \dots + \BI \Big)\\
&= \BK^r\Big(\sum_{j=0}^{q-1} \BK^{hj} \Big)
\Big(\BK^{h-1}+ \BK^{h-2} +\dots + \BK+\BI \Big)  + \Big( \BK^{r-1} +
\dots + \BI \Big)
\end{align*}
We may ignore the last term since multiplying it by $\dfrac{\rho^h-1}{\rho^{hq}-1}$ and
taking the limit as $q$ goes to infinity gives 0.
Recall that $\BK^h$ is a block-diagonal matrix with entries $\BK_1\cdots
\BK_h$, $\BK_2\cdots \BK_h\BK_1$, and so on to
$\BK_h\BK_1\cdots\BK_2$.  
Each of these matrices is primitive and each has
$\rho^h$ as its largest eigenvalue.  As in the
proof of Theorem~\ref{t:primitive}, where $\BK$ is primitive, 
large powers of these matrices
are approximated using the product of Perron vectors.  Thus we have 
\begin{align*}
\lim_{q \rta \infty}\dfrac{\rho^{h}-1}{\rho^hq-1}\sum_{j=0}^{q-1}
\BK^{hj} 
& =
\begin{bmatrix}
\Bz_1\Bystar_1 & \Bzero & \Bzero & \Bzero & \dots & \Bzero \\
\Bzero & \Bz_2\Bystar_2 & \Bzero & \Bzero &  \dots & \Bzero \\
\Bzero & \Bzero &\Bz_3\Bystar_3 & \Bzero &  \dots & \Bzero \\
\dots & \dots&\dots &\dots &\dots&\Bzero \\
\Bzero & \Bzero & \Bzero & \Bzero & \dots & \Bz_h\Bystar_h 
\end{bmatrix}
\end{align*}
It is readily checked that $\BK^{h-1}+ \BK^{h-2} +\dots + \BK+ \BI$ is
equal to
\begin{align*}
\begin{bmatrix}
\BI & \BK_1 & \BK_1 \BK_2 & \BK_1\BK_2\BK_3 & \dots & \BK_1\BK_2 \dots \BK_{h-1} \\
\BK_2\dots \BK_{h-1}\BK_h & \BI & \BK_2 & \BK_2\BK_3 & \dots & \BK_2\dots \BK_{h-1}\\
\BK_3 \dots \BK_{h-1}\BK_h & \BK_3 \dots \BK_{h-1}\BK_h \BK_1& \BI & \BK_3 & \dots &
\BK_3\dots \BK_{h-1}\\
\dots & \dots&\dots &\dots &\dots& \dots \\
\BK_h  & \BK_h\BK_1 & \BK_h\BK_1\BK_2 & \BK_h\BK_1\BK_2 \BK_3 &  \dots & \BI\\
\end{bmatrix}
\end{align*}
Computing the product of these last two expressions and simplifying
using the formulas for $\Bystar_i$, we arrive at the desired formula
for $r=0$.

\begin{align*}
\lim_{q\rta \infty}\dfrac{\rho^h-1}{\rho^{hq}-1}\sum_{i=0}^{hq-1} \BK^i 
&=\begin{bmatrix} 
\Bz_1\Bystar_1 & \rho \Bz_1 \Bystar_2  &\rho^2 \Bz_1 \Bystar_3  & \rho^3 \Bz_1
\Bystar_4  & \dots &\rho^{h-1} \Bz_1 \Bystar_h  \\
\rho^{h-1} \Bz_2 \Bystar_1  & \Bz_2\Bystar_2 &\rho  \Bz_2\Bystar_3  &\rho^2
\Bz_2\Bystar_4  &  \dots &\rho^{h-2} \Bz_2 \Bystar_h  \\ 
\rho ^{h-2}\Bz_3\Bystar_1  & \rho ^{h-1} \Bz_3 \Bystar_2  &\Bz_3\Bystar_3 &
\rho \Bz_3 \Bystar_4  &  \dots & \rho^{h-3} \Bz_3 \Bystar_h  \\ 
\dots & \dots&\dots &\dots &\dots&\dots  \\
\rho \Bz_h \Bystar_1  &\rho^2 \Bz_h \Bystar_2  &\rho ^3 \Bz_h \Bystar_3
&\rho^4 \Bz_h \Bystar_4  & \dots & \Bz_h\Bystar_h  
\end{bmatrix} \\
&= \BZ\BTheta \BY 
\end{align*}

It is readily verified that $\BK\BZ = \rho \BZ \BP$,
so $\BK^r\BZ = \rho^r\BZ \BP^r$.
Thus we have for any~$r$,
\begin{align*}
\lim_{h\rta \infty}\dfrac{\rho^h-1}{\rho^{hq}-1}
\BK^r \sum_{i=0}^{hq-1} \BK^i 
&= \BK^r \BZ\BTheta \BY \\
&= \rho^r \BZ \BP^r \BTheta \BY
\end{align*}
\end{proof}

For the following theorem, let $\BLam_i$ be the submatrix of $\BLam$
with rows indexed by elements of $E_i$ and let $\Bc_i =
\Bystar_i\BLam_i$, so that
\begin{align*}
\begin{bmatrix}
\Bc_1 \\
\Bc_2 \\
\Bc_3 \\
\dots \\
\Bc_h 
\end{bmatrix}
= 
\begin{bmatrix}
\Bystar_1 & \Bzero & \Bzero &\dots & \Bzero \\
\Bzero & \Bystar_2 & \Bzero  &\dots & \Bzero \\
\Bzero & \Bzero & \Bystar_3  &\dots & \Bzero \\
\dots &  \dots &  \dots &  \dots &  \dots \\
\Bzero & \Bzero & \Bzero & \dots &\Bystar_h 
\end{bmatrix}
\begin{bmatrix}
\BLam_1 \\
\BLam_2 \\
\BLam_3 \\
 \dots \\
\BLam_h
\end{bmatrix}
\end{align*}
For notational convenience, we define $\Bc_i$ for any integer by
reducing the index modulo $h$ using representatives $1,2,\dots,h$.

\begin{theorem}
\label{t:imprimitive}
For $r \in \set{0,\dots,h-1}$ and $e \in E_i$, 
\[
\lim_{q \rta \infty}\left( x_e ^{(hq+r)}\right) ^{\frac{1}{\rho^r}\cdot\frac{\rho^h -1}{\rho^{hq}-1}}
= \Big( \prod_{j=0}^{h-1} \left(\Bu^{\Bc_{r+i+j}}\right)^{\rho^j}\Big)^{\Bz_e } 
\]

The sum-product algorithm converges if and only if the following 
products have the same parity as $r$ varies.
\[
U_r = \prod_{j=0}^{h-1}\left(\Bu^{\Bc_{r+j}}\right)^{\rho^j}
\]

\end{theorem}

\begin{proof}
Dividing the equation of Lemma~\ref{l:not primitive} by $\rho^r$ and
multiplying by $\BLam$ we have 
$\frac{1}{\rho^r}\cdot\frac{\rho^h-1}{\rho^{qh} -1} \BA^{(qh+r)}
\approx \BZ\BP^r\BTheta\BY\BLam$. 
The matrix $\BY\BLam$ is $h\times \abs{L}$ with $i^{\rm th}$ row $\Bc_i$, so the
$i^{\rm th}$ row of $\BTheta\BY\BLam$ is 
$\Bc_i + \Bc_{i+1}\rho+ \Bc_{i+2}\rho^2+\dots+\Bc_{i+h-1}\rho^{h-1}$ where
subscripts are computed modulo $h$ using representatives $1,\dots,h$.
The $i^{\rm th}$ row of $\BP^r\BTheta\BY\BLam$ is 
$\Bc_{r+i} + \Bc_{r+i+1}\rho+
\Bc_{r+i+2}\rho^2+\dots+\Bc_{r+i+h-1}\rho^{h-1}$.
Thus if $e \in E_i$ then 
\begin{align*}
\lim_{q \rta \infty} \frac{1}{\rho^r} \frac{\rho^h -1}{\rho^{hq-1}} \Ba_e ^{(hq+r)}
 &= \Bz_e \sum_{j=0}^{h-1} \Bc_{r+i+j}\rho^j\\
\intertext{and}
\lim_{q \rta \infty}\left( x_e ^{(hq+r)}\right) ^{\frac{1}{\rho^r}
\frac{\rho^h -1}{\rho^{hq-1}}}
&= \Big( \prod_{j=0}^{h-1} \left(\Bu^{\Bc_{r+i+j}}\right)^{\rho^j}\Big)^{\Bz_e } 
\end{align*}

For convergence of $x_e$,  the value of $\Bz_e$ is irrelevant.
For $i=0$  we get the condition stated in the Theorem, but for other
$i$ we get the same set of products since the subscripts on $\Bc_i$ are
computed modulo $h$.
\end{proof}

\section{Examples}\label{s:examples}
\begin{example}
Let $G$ be $d$-regular with $n$ vertices and therefore $nd$ edges. 
Then $\tildeG$ and its adjacency matrix $\BK$ are $(d-1)$-regular.  Let
us assume that $\BK$ is primitive. Then the Perron eigenvalue of $\BK$ is 
$\rho=d-1$ and the Perron vectors, of length $nd$, are constant, that
is, multiples of $[1,1,\dots, 1]$ or its transpose.
We may take  $\Bz$ to be 1 in each component and $\Bystar$ to be $1/nd$
in each component.  The matrix $\BLam$ is $nd\times n$ and has $d$ ones in
each column, so $\Bc=\Bystar\BLam = [ 1/n, 1/n, \dots, 1/n]$.  
Theorem~\ref{t:primitive} says that the SPA 
converges to 0 when $\Big(\prod_{\ell \in L}u_\ell\Big)^{1/n} <1$, or,
equivalently, when $\prod_{\ell \in L}u_\ell <1$.  The SPA converges to 
$\infty$ when $\prod_{\ell \in L}u_\ell >1$.  The convergence is roughly
exponential with base $\prod_{\ell \in L}u_\ell >1$ and 
exponent $\rho/n$.
\end{example}

\begin{example}
Consider the graph $G$  and its flow graph $\tilde G$ below.
There are 3~vertices, $A,B,C$ and 4~conjugate pairs of edges in $G$;
one edge of each pair is shown.
\[
\xymatrix{
A \ar@/^1pc/[rr]^1
\ar@/_1pc/[rr]_2  \ar@{->}[ddr]_3 & 
& B \ar@{<-}[ddl]^4 \\
&& \\
 &  C &
}
\qquad 
 \xymatrix{
 &1\ar@/_1pc/[rrr] &&&{\bar 2} \ar@/_1pc/[lll]&\\
 {\bar 3}\ar@{->}[ur]\ar@{->}[dr]  & &
\bar 4\ar@{<-}[ul] \ar@{<-}[dl] \ar@{->}[ll] &
 3\ar@{->}[rr] \ar@{<-}[dr] \ar@{<-}[ur]   &&
 4 \ar@{->}[dl] \ar@{->}[ul] 
\\
 &2\ar@/_1pc/[rrr]&&&{\bar 1} \ar@/_1pc/[lll]&\\
 }
\]
\bigskip
The matrix $\BK$ is $8\times 8$ and its largest eigenvalue, $\rho$,
is the positive root of $x^3-x^2-2$.
Taking the edges in the order
$1, 2, 3,4, \bar 1, \bar 2, \bar 3, \bar 4$, the eigenvector $\Bystar $
is, up to multiple, 
$[1, 1, \rho-1, \rho^2-\rho, 1, 1,\rho^2-\rho, \rho-1]$.
Taking the vertices in alphabetical order, the transpose of $\BLam$ is
\[\BLam ^T = 
\begin{bmatrix}
1 & 1 & 1 &0   & 0 & 0 &0 & 0  \\
0 & 0 & 0 &0   & 1 & 1 &0 & 1  \\
0 & 0 & 0 &1   & 0 & 0 &1 & 0  
\end{bmatrix}
\]
so $\Bc = \Bystar \BLam = [\rho+1, \rho+1, 2(\rho^2-\rho)]$.
Since $\rho\approx 1.6956$, we have
$\Bystar \approx [1,1, 0.7, 1.2, 1, 1, 1.2, 0.7]$ and 
$\Bc \approx [2.7, 2.7, 2.4] $.
Not surprisingly, based on the higher degree of the nodes $A$ and $B$,
the input values for $A$ and $B$ have a
greater influence on convergence than does $C$.
Notice that it is possible for the SPA to return the codeword 000 when
111 is most likely.  Such is the case for  $u_A=u_B = 0.5$, and $u_C =
4.5$.  
\end{example}

The last example might lead one to suppose that higher-degree nodes
have a greater influence on 
convergence than lower-degree nodes---that is, if $\ell$ has higher
degree than $\ell'$, then $\Bc_\ell > \Bc_{\ell'}$---but this is not
necessarily  the case.  

\begin{example}
Consider the graph below. Up to scaling, 
\[\Bc\approx  [2.1, 2.1, 1.7, 1.2, 1.8, 1.3, 1.3, 1.1, 1.6, 1.6,
1.6, 1.6 ] \]
Although the node at the top has higher degree than the nodes at the
bottom,  the corresponding component in $\Bc$ is smaller; {\em e.g.}, $\Bc_4 <
\Bc_8$.

\[
 \xymatrix{
&&4&& \\
5 \ar@{<->}[urr] \ar@{<->}[d]   && 7 \ar@{<->}[u]
\ar@{<->}[d]
&& 6 \ar@{<->}[ull] \ar@{<->}[d]\\
1 && 3\ar@{<->}[ll] \ar@{<->}[rr] && 2 \\
8\ar@{<->}[u] \ar@{<->}[urrrr]& 9 \ar@{<->}[ul]
\ar@{<->}[urrr] && 10 \ar@{<->}[ulll] \ar@{<->}[ur] & 
11 \ar@{<->}[u] \ar@{<->}[ullll]
}
\]
\end{example}

\begin{example}
 Let $G$ be bipartite, so that $L$ is the disjoint union of $L_1$ and
$L_2$ and $E$ is the disjoint union of $E_1$ and $E_2$ with all edges
in $E_i$ having source in $L_i$.  Conjugation gives a bijection of
$E_1$ with $E_2$.  Enumerating the edges of $E_1$ before those of
$E_2$ and the vertices of $L_1$ before those in $L_2$ we have 
\[
\BK= \begin{bmatrix}
0 & \BK_1\\
\BK_2 & 0 
\end{bmatrix}
\qquad
\BLam= \begin{bmatrix}
\BLam_1 & 0\\ 
 0  & \BLam_2
\end{bmatrix}
\]
Notice that $\BLam_i$ differs slightly from the notation used in the
theorem, where $\BLam_i$ is the upper $\abs{E_1}$ rows of $\BLam$. 

Let us assume that the index of imprimitivity of $\BK$ is $h=2$.
Suppose that $G$ is $(d_1, d_2)$-regular and that $n_i=\abs{L_i}$, so
that $\abs{E_i} = n_id_i$, and $n_1d_1=n_2d_2$.
Each edge in $E_1$ has $d_1-1$ non-conjugate edges which feed
into it, so  $\BK_1$ is $(d_1-1)$-regular. 
Similarly,  $\BK_2$ is $(d_2-1)$-regular.
Each row and each column of the matrices $\BK_1\BK_2$ and $\BK_2\BK_1$
sums to $(d_1-1)(d_2-1)$, so  $\BK_1\BK_2$ and $\BK_2\BK_1$ have
Perron eigenvalue  $(d_1-1)(d_2-1)$ and Perron eigenvectors that are
constant.
The Perron eigenvalue of $\BK$ is  $\rho = \sqrt{(d_1-1)(d_2-1)}$.
We take $\Bystar_1$ to be $1/\sqrt{d_1-1}$ in each component.
Then  $\Bystar_2 = \rho^{-1}\Bystar_1 \BK_1$ is $1/\sqrt{d_2-1}$ in
each component.   

Since $\BLam_i$ has $d_i$ 1's in each column, $\Bystar_i\BLam_i $ is 
$d_i/\sqrt{d_i-1}$ in each component. 
In the notation of  Theorem~\ref{t:imprimitive} we have
\begin{align*}
(\Bc_i)_\ell&= \begin{cases}
\frac{d_i}{\sqrt{d_i-1}} & \text{ if } \ell\in E_i  \\
0 & \text{ otherwise} \\
\end{cases} \\
U_0 &=  \Bu^{\Bc_2}\left(\Bu^{\Bc_1}\right)^\rho 
= \Big( \prod_{\ell \in L_2} u_\ell\Big) ^{\frac{d_2}{\sqrt{d_2-1}}}
\Big(\prod_{\ell \in L_1} u_\ell\Big) ^{d_1 \sqrt{d_2-1}}  \\
U_1 &= \Bu^{\Bc_1}\left(\Bu^{\Bc_2}\right)^\rho 
= \Big(\prod_{\ell \in L_1} u_\ell\Big)^{\frac{d_1}{\sqrt{d_1-1}}}
\Big( \prod_{\ell \in L_2} u_\ell\Big)^{d_2\sqrt{d_1-1}} 
\end{align*}
Notice that for $U_0$ the ratio of the exponents appearing in the
formula is $d_1(d_2-1)/d_2$, while for $U_1$ it is 
$d_2(d_1-1)/d_2$.  
The SPA converges if and only if $U_0$ and $U_1$ have the same parity.

If $G$ is regular,  $d_1=d_2=d$, then  $n_1=n_2$ and  $\rho= d-1$.
The SPA converges if and only if 
$ \Big(\prod_{\ell \in L_1} u_\ell\Big) \Big(\prod_{\ell \in L_2} u_\ell\Big)^{d-1}$
and
$ \Big(\prod_{\ell \in L_1} u_\ell\Big)^{d-1} 
\Big(\prod_{\ell \in L_2} u_\ell\Big)$
have the same parity.

\end{example}

It is possible to have higher index of primitivity; 
Figure~\ref{f:imprimitivity} gives several examples.

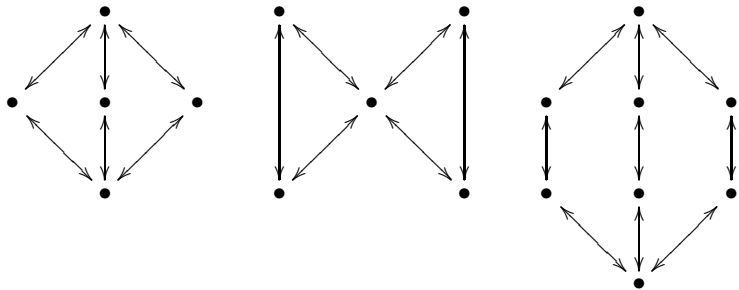
\begin{figure}[hbt]
\label{f:imprimitivity}
\[
\xymatrix{
& \bullet\ar@{<->}[dl]\ar@{<->}[d]\ar@{<->}[dr] & \\
\bullet & \bullet & \bullet\\
& \bullet\ar@{<->}[ul]\ar@{<->}[u]\ar@{<->}[ur] & 
}
 \qquad 
\xymatrix{
\bullet \ar@{<->}[dd] && \bullet \ar@{<->}[dd] \\
& \bullet\ar@{<->}[ul]\ar@{<->}[dl]\ar@{<->}[ur] \ar@{<->}[dr] & \\
\bullet && \bullet
}
\qquad
\xymatrix{
& \bullet\ar@{<->}[dl]\ar@{<->}[d]\ar@{<->}[dr] & \\
\bullet \ar@{<->}[d] &\bullet  \ar@{<->}[d] & \bullet \ar@{<->}[d]\\
\bullet &\bullet& \bullet\\
& \bullet\ar@{<->}[ul]\ar@{<->}[u]\ar@{<->}[ur] & 
}
\]

\caption{Three graphs $G$ such that the index of imprimitivity $h$  of the
  adjacency matrix of $\tildeG$ is larger than $2$.  From left to
  right, $h=4$, $h=3$, $h=6$.}
\end{figure}

\newpage
\section{Trapping Sets}
\label{s:trap}

This section extends results of the previous sections to trapping
sets.
An $(a,b)$ {\it trapping set }  is a subgraph of a bipartite graph
consisting of  $a$ bit nodes
and all their neighboring check nodes such that  $b$ check
nodes on the subgraph have odd degree.  
Richardson introduced the term ``trapping set" in
\cite{Rich03}, where he used a combination of simulation and combinatorial analysis
of trapping sets to predict error floors.
MacKay and Postol used the term {\it near-codeword} for the same
phenomena in \cite{MacPo03}, since for small $b$, a vector with
support on the $a$ bit nodes is ``nearly'' a codeword.
McKay and Postol attributed decoding failure 
of a Margulis code in the error floor region to near-codewords.
 A great deal of research has been focused on trapping sets.
 An edge swapping construction that  eliminates small trapping sets was shown 
 to lower the  error floor in \cite{Vasic08,HanRyan09}. 
Analysis of trapping sets  led to accurate prediction of decoding
 performance of LDPC codes  on the binary symmetric channel using the
 Gallager~A algorithm in \cite{VasicFloor}. 
Counts of small trapping sets were shown to be good predictors for performance of the SPA in \cite{LampBrOS}. 
Planjery et al  \cite{VasicBeyond,VasicMultilevel} developed a message-passing
algorithm for the binary symmetric channel that is similar to belief
propagation, but is designed to overcome errors due to trapping sets.

We are going to identify the conditions under which the SPA converges
on a trapping set, but first we must  properly frame the problem.
As Richardson and others have noted (see for example the ``trapping set
ontology" of \cite{VasicOntology}) simulation has shown that the
trapping sets that affect the error floor invariably have check nodes
of degree at most~2.
The graph obtained by removing the degree-$1$ checks from such a trapping set is either a cycle or
one of the graphs treated in Section~\ref{s:convergence};
we will call this graph the {\it core} of the trapping set.
Note that considering the SPA on a detached trapping set yields nothing useful: In the SPA, a check of degree 1 causes all messages from the
neighboring bit to other checks to be 0 after the first iteration of the
algorithm, and this  causes the SPA to eventually converge to
the zero codeword on the trapping set. 
To avoid this,  we can modify the  trapping set by adding a bit node to each 
check node of degree~1.  This ``virtual bit'' serves as the
communication link between the ambient graph and the core graph.
Starting with an $(a,b)$ trapping set we now have a bipartite graph with 
all check nodes of degree~2 and $b$ bit nodes of degree~1.
It will ease our analysis to add a new degree-1 bit node 
{\it for each bit} of our core graph, along with a check node connecting the
two.  This might seem to violate the essential interest in trapping
sets, the small number $b$ of odd degree check nodes, but, as we now
show,  the  behavior of the
SPA on the trapping set may be easily deduced from our somewhat larger
graph.

Consider for a moment the SPA on an arbitrary bipartite graph $B'$ having
a bit node $\ell'$ of degree~1.  It is straightforward to verify that if $u_\ell'=1$, the
SPA on $B'$ has the same edge messages
at any iteration $t$---as a function of the $u_{\ell}$ for $\ell\in L$---as the
SPA on the graph $B$ obtained from $B'$ by removing $\ell'$ and its incident edge.
In other words, setting $u_{\ell'}=1$, for $\ell'$ a leaf
node, effectively removes $\ell'$ from the algorithm. 
Seen another way, suppose we start with a bipartite graph $B$ and add a bit
node $\ell'$ and an edge to create a new graph $B'$.  The
 behavior of the
SPA on $B$  can be recovered from the behavior of the SPA on $B'$ by
setting $u_{\ell'}=1$.  

Our strategy to understand the behavior of the SPA on trapping sets is therefore
the following: we start with a bipartite graph 
$B= (E,L,R,\lam,\rho)$ satisfying the
properties of Section~\ref{s:convergence}, and the associated
undirected graph $G$. 
We let $\BLam$ and $\BK$ be the edge-vertex incidence matrix and flow
matrix, for $G$.
Let $L= \{\ell_1,\dots \ell_n\}$; now, create a new set
of bit nodes $L' = \{\ell'_1,\dots, \ell'_m\}$ along with new edges 
$E_1 = \{e_1, \dots, e_n\}$ and $E_2 = \{\bar e_1, \dots, \bar e_n\}$ so that
edge $e_i$ goes from $\ell_i'$ to $\ell_i$, and $\bar e_i$ is its 
conjugate edge going from $\ell_i$ to $\ell_i'$. 
Thus, we have effectively attached a leaf node to each node of $G$ to
create a new graph; call this graph $G'$ and the associated bipartite
graph $B'$. 
Since all check nodes have degree~2, the SPA passes monomials in the
input values $u_\ell$, so may analyze the SPA on the graph we have constructed via the local sum algorithm on the associated undirected graph.  
The main result is the following.

\begin{theorem}
\label{t:trapset}
With the preceding notation, the sum-product algorithm on $B'$
with inputs $u_i$ for $\ell_i \in L$ and $u'_{i}$ for $\ell_i'\in L'$  converges  if and only if 
the sum-product algorithm on $B$ with input $u'_iu_i^\rho$ converges.
Here $\rho$ is the Perron eigenvalue for the flow matrix of $B$.
\end{theorem}

Our primary interest in this theorem is the application to trapping
sets, in which case we set $u_i '=1$ for any $\ell$ which is not
connected to a check node of degree~1 in the trapping set.

\begin{proof}
The edge-vertex matrix for our new graph is
\[
\BLam' = \begin{bmatrix}
\BI & 0 \\
0 & \BLam\\
0 & \BI
\end{bmatrix}
\]
The flow matrix $\BK'$ is a $3\times 3$ block matrix, corresponding to
the partition of the edges into $E_1$, $E$ and $E_2$.
The rows for  $e_i \in E_1$ are 0, because
$\lam(e_i)$ is a leaf, so no edge terminates at $\lam(e_i)$
except $\bar e_i$.
Similarly, columns for  $\bar e_i \in E_2$   are 0, since the edges
in $E_2$ terminate in a leaf.  
The  block matrix for  $E_2 \times E_1$ is 0, since we don't allow
flow into conjugate edges.
Since $e_i \in E_1$  terminates in $\ell_i$, $e_i$ flows into the
edges of $E$ that have source $\ell_i$. 
Thus the block matrix for $E\times E_1$ is
exactly the same as the edge-vertex incidence matrix for $G$,
namely, $\BLam$.
The matrix for $E_2\times E$ is similar, 
the edges that flow into $\bar e_i$, are the edges in $E$ that
terminate at $\ell_i$.
Let $\BT$ be the  $|E|\times |E|$ matrix for defined by 
$\BT_{e,\bar{e}}=1$ and $T$ is 0 elsewhere, so $\BT^2=\BI$.
The $E_2 \times E$ matrix is $\BLam\tr \BT$. 
Finally, the $E \times E$ portion of $\BK'$ is $\BK$ the flow matrix
for $G$. Thus we have

\[\BK' = 
\begin{bmatrix}
0 & 0 & 0 \\
\BLam  & \BK & 0\\ 
0  & \BLam\tr \BT & 0
\end{bmatrix}
\]

The matrix $\BT$ defined gives a useful expression for the flow matrix of the core graph, and for its Perron eigenvector $\Bystar$.  It is readily verified that 
$\BK = \BLam \BLam\tr \BT -\BT$.    Then $\Bz\tr \BT$ is a left eigenvector for $\BK$:
\begin{align*}
( \BLam \BLam\tr\BT -\BT)\Bz &=\rho \Bz, \quad \text{so} \\
\Bz\tr  (\BT \BLam \BLam\tr -\BT)  &= \rho \Bz\tr
\intertext{Multiplying on the right by $\BT$ and noting that $\BT^2= \BI$,}
\Bz\tr \BT  (\BLam \BLam\tr\BT -\BT)  &= \rho \Bz\tr \BT
\end{align*}
Thus we may take $\Bystar = \Bz\tr \BT$ provided we normalize $\Bz$ so that $\Bz\tr\BT\Bz=1$.

Consider the case when $\BK$ is primitive (the imprimitive  case is similar).  
Using the approximation of 
Section~\ref{s:convergence}, for $t>1$, 
\begin{align*}
(\BK')^t &=
\begin{bmatrix}
0 & 0 & 0 \\
\BK^{t-1}\BLam  & \BK^t & 0\\ 
\BLam\tr \BT\BK^{t-2} \BLam  & \BLam\tr \BT\BK^{t-1} & 0
\end{bmatrix}\\
&\approx 
\begin{bmatrix}
0 & 0 & 0 \\
\rho^{t-1}\Bz\Bystar \BLam  & \rho^t z\Bystar & 0\\ 
\rho^{t-2}\BLam\tr \BT z\Bystar \BLam  & \rho^{t-1}\BLam\tr \BT z\Bystar & 0
\end{bmatrix}
\\
&= \begin{bmatrix}
0 & 0 & 0 \\
\rho^{t-1}\Bz\Bz\tr  \BT  \BLam  & \rho^t \Bz\Bz\tr  \BT  & 0\\ 
\rho^{t-2}\BLam\tr  \BT  \Bz\Bz\tr  \BT \BLam  & \rho^{t-1}\BLam\tr  \BT  \Bz\Bz\tr  \BT  & 0
\end{bmatrix}
\end{align*}

Now compute 
\begin{align*}
(\BA')^{(t)} &= \Big( \sum_{i=0}^{t-1} (\BK')^i \Big) \BLam' \\
&\approx  \begin{bmatrix}
\BI & 0 \\
\frac{\rho^{t-1}-1}{\rho-1} \Bz\Bz\tr\BT \BLam& \frac{\rho^t
  -1}{\rho-1}\Bz\Bz\tr  \BT \BLam\\
\frac{\rho^{t-2}-1}{\rho-1}\BLam\tr  \BT \Bz\Bz\tr \BT \BLam  & 
\frac{\rho^{t-1}-1}{\rho-1}\BLam\tr  \BT  \Bz\Bz\tr\BT \BLam 
\end{bmatrix}
\end{align*}

Let $\Bc = \Bystar \BLam =\Bz\tr\BT\BLam$ (a $1 \times n$ vector, since $n=|L|$).
This is  the same vector we used for the core graph.  Taking the limit,
\begin{align*}
\lim_{t \rta \infty} \dfrac{\rho -1}{\rho^{t-1}-1} \BA^{(t)}
& = \begin{bmatrix}
0 & 0 \\
\Bz \Bc & \rho \Bz \Bc \\
\rho^{-1} \Bc\tr \Bc  & \Bc\tr \Bc & 0
\end{bmatrix} 
\end{align*}
Now let $\Bu$ be the vector of variables for $L$ and $\Bu'$ the vector
of variables for $L'$.  We have for $e \in E$
\begin{align*}
\lim_{t \rta \infty} \dfrac{\rho -1}{\rho^t-1} \Ba_e^{(t)} &=
\begin{bmatrix}
\Bz \Bc & \rho \Bz \Bc
\end{bmatrix} \\
\lim_{t \rta \infty}\left( x_e ^{(t)}\right) ^{\frac{\rho -1}{\rho^t-1}}
&= 
\lim_{t \rta \infty} 
\Big( (\Bu')^{\Bc}\Bu^{\rho\Bc }\Big)^{\Bz_e }\\
\intertext{For $\bar{e_i}\in E_2$ recall that $\bar{e_i}$ is the edge from $\ell_i\in L$ to 
$\ell_i' \in L'$}
\lim_{t \rta \infty}\left( x_{\bar{e_i}} ^{(t)}\right) ^{\frac{\rho -1}{\rho^t-1}}
&= 
\lim_{t \rta \infty} 
\Big( (\Bu')^{\rho^{-1}\Bc}\Bu^{\Bc }\Big)^{\Bc_i }
\end{align*}

For any edge $e$, $x_e$ goes to 0 if $(\Bu')^\Bc\Bu^{\rho\Bc}<1$ and to 
$\infty$ if $(\Bu')^{\Bc}\Bu^{\rho\Bc} >1$. 
Noting that $(\Bu')^\Bc \Bu^{\rho\Bc} =\prod_{\ell\in L} (u'_\ell(u_\ell)^\rho)^{\Bc_\ell}$
establishes the theorem.
\end{proof}

\begin{example}
Consider each of  the graphs below 
as the core
of a trapping set in which all bit nodes have degree~3.
Consider the bit nodes of degree~2 in the figure as attached  to another
bit node, which is not shown.  All three appear in the trapping set ontology \cite{VasicOntology}, and the first graph is a prominent example in the papers of Planjery, Vasic and co-authors.
\[
\xymatrix{
& \bullet\ar@{<->}[dl]\ar@{<->}[d]\ar@{<->}[dr] & \\
\bullet & \bullet & \bullet\\
& \bullet\ar@{<->}[ul]\ar@{<->}[u]\ar@{<->}[ur] & 
}
\quad
\xymatrix{
& \bullet\ar@{<->}[dl]\ar@{<->}[dr] & \\
\bullet \ar@{<->}[d]  \ar@{<->}[rr]&& \bullet \ar@{<->}[d]\\
\bullet \ar@{<->}[rr]&& \bullet\\
}
\qquad 
\xymatrix{
\bullet \ar@{<->}[dd]  \ar@{<->}[r] &\bullet \ar@{<->}[dd]  \ar@{<->}[r] &\bullet \ar@{<->}[dd] \\
&&\\
\bullet  \ar@{<->}[r] &\bullet   \ar@{<->}[r] &\bullet \\
}
\]
The first two graphs  lead to $(5,3)$ trapping sets and the third
graph leads to a $(6,2)$ trapping set.    
The flow matrix for the $(5,3)$ with girth~8 has index of imprimitivity~4 with
$\rho=\sqrt{2}$.  The flow matrix for the $(5,3)$ with  girth~6 is
primitive, since the gcd of admissible cycles is~1,  and $\rho=1.424$, slightly larger than $\sqrt{2}$.
The flow matrix for the $(6,2)$ graph has index of imprimitivity~2,
with $\rho=1.353$.

The first graph is particularly interesting.
Label the edges from the top node down as $1,2,3$ and from the bottom node up as $4,5,6$, in each case starting from the left to the right.  Order the edges as follows:  $1,2,3,\bar{1},\bar{2}, \bar{3}, 4,5,6,\bar{4},\bar{5},\bar{6}$.  Let $\BI$ be a $3\times 3$ identity matrix, let $\Bone$ be a $3\times 1$ vector of with one in each entry, and let $\BJ = \Bone\Bone\tr$. 
Let $\Bzero$ be a zero matrix or zero vector as determined by context.
With an appropriate ordering for $L$, we get the following structural matrices.
\begin{alignat*}{2}
\BLam &= \begin{bmatrix}
\Bzero & \Bone & \Bzero \\
\BI & \Bzero& \Bzero\\
\Bzero &  \Bzero & \Bone \\
\BI & \Bzero& \Bzero\\
\end{bmatrix} & \qquad
\BT &= \begin{bmatrix}
\Bzero & \BI & \Bzero &\Bzero \\
\BI & \Bzero& \Bzero&\Bzero\\
\Bzero &  \Bzero & \Bzero &\BI \\
 \Bzero& \Bzero&\BI&\Bzero
\end{bmatrix} 
\\
\BLam\tr \BT &= \begin{bmatrix}
\BI & \Bzero& \BI & \Bzero\\
\Bzero &\Bone \tr&  \Bzero & \Bzero \\
\Bzero& \Bzero & \Bzero &\Bone\tr
\end{bmatrix} & \qquad
\BK &= \begin{bmatrix}
\Bzero & \BJ-\BI & \Bzero &\Bzero \\
\Bzero& \Bzero& \BI &\Bzero\\
\Bzero &  \Bzero & \Bzero &\BJ-\BI \\
\BI & \Bzero& \Bzero&\Bzero
\end{bmatrix}
\end{alignat*}

One might expect that the $(5,3)$ trapping set with girth~8 performs
better under the sum-product algorithm than the $(5,3)$ trapping set
with girth~6, but the reverse is true, as can be seen from the
performance curves in Figure~\ref{f:trap(5,3)}.   
The primitivity of the flow matrix of the girth~6 graph makes it less susceptible to channel noise.
\end{example}

\begin{figure}[htb]
\centering
\epsfig{figure=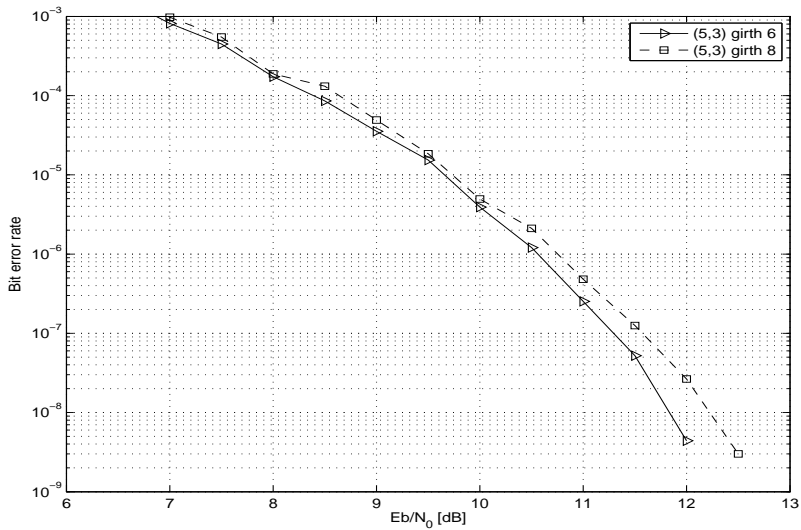,height=2.8in, width=4.2in, angle=0}
\caption{\label{f:trap(5,3)}
Comparison of the performance of the SPA on the core graph of  two
$(5,3)$ trapping sets.  The termination criterion was $\eps = 10^{-8}$.}
\end{figure}

We conclude this section with a suggested program that may result in  
closed-form probability calculations for the  decoding failure on any
graph, analogous to the results of Chilappagari et al~\cite{VasicFloor}
on the binary symmetric channel (see also \cite{VasicBeyond,VasicMultilevel}). 
A key to their analysis of decoding algorithms is an ``isolation
assumption,'' which gives conditions under which the behavior of the decoding
algorithm on the trapping set can be  isolated from  the behavior
on the rest of the graph.   The essence of the isolation assumption is
that for some length $M$ there is no  path through the complement of
the trapping set,  between nodes of the trapping set.  Thus for $M$
iterations of the algorithm,  the affect of one edge message of the
trapping set on another edge message of the trapping set is completely
determined by the message passing of the trapping set and unaffected
by paths in the complement.

In order to simplify the discussion and make definite statements, we
will assume that all bits in the original graph have degree~$3$. Each
bit in the core of a trapping set has degree $2$ or $3$, so it is
``missing" at most one check. When we add in the new bits $\ell'$ to
the core, the core bits now have degree $3$ or $4$. For those of
degree $4$, the check we added does not correspond to a check in
actual graph, so set the $u_{\ell'}$ to $1$; the remaining $\ell'$ --
actually, their attendant checks -- can be thought of as representing
the incoming information to the trapping set from the rest of the
graph. Under the assumption, then, that these incoming messages are
essentially independent of the outgoing messages passed out of the
trapping set, one should be able to obtain an expression for the
probability that some bit is in error after the $t^{th}$ iteration of
the algorithm via the initial distribution on the core bits,
density-evolution analysis on the check-to-bit messages from the new
checks for which $u_{\ell'}$ was not set equal to $1$ (either via the
methods of~\cite{RichUrb:Book} or via Monte Carlo simulation), and the
formula in Theorem~\ref{t:trapset}.

\section{Concluding Remarks}
In this section we summarize our results on regions of convergence, we
comment on covering graphs, and we present some experimental results
illustrating  the preceeding theory.

\subsubsection*{Region of Convergence}
When the matrix $\BK$ is primitive, Theorem~\ref{t:primitive} shows 
that the SPA will converge except on a
region of measure 0, defined by the equation $\Bu^{\Bc}=1$.
For regular graphs, the SPA is therefore a maximum likelihood decoder, since
this hypersurface is $\prod_{\ell \in L} u_\ell=1$, which corresponds to
$\prod_{\ell\in L} p_\ell(1) = \prod_{\ell \in L} p_\ell(0)$.  For an irregular
graph, the components of $\Bc$ differ, and the SPA may place more
emphasis on certain $u_\ell$.   

When the matrix $\BK$ is imprimitive, there is a region of positive
measure on which the SPA does not converge.  The simplest example is
$K_{3,2}$, the complete bipartite graph on $3$ bits and $2$ checks.  
The SPA  will converge to the all-$0$ codeword on the region
defined by $u_1^2u_2<1$ and $ u_1u_2^2<1$ and to the all-$1$ codeword
on the region $u_1^2u_2>1$ and $ u_1u_2^2 >1$.
The SPA will not
converge on the region between the curves $u_1u_2^2=1$ and
$u_1^2u_2=1$. 

The secondary eigenvalues may also affect performance.
Suppose we use the criterion for  termination of the SPA 
defined ealier:  decode to the $0$-codeword if for all $\ell \in L$,
$\hat{u}_\ell <\eps$ and to the all-$1$-codeword if for all $\ell\in
L$, $\hat{u}_\ell> 1/\eps$. 
If the secondary eigenvalues  are close to the Perron eigenvalue  or if there are
many large secondary eigenvalues,  it may take a longer time for the
Perron vector to become dominant.  
Since $\hat{u}_\ell$ involves large powers of the $u_\ell$ (as $t$
gets large)  it tends toward extremes ($0$ or $\infty$),
and this may cause the algorithm to terminate early and incorrectly.

\subsubsection*{Covering Graphs}
It is common to construct LDPC codes by using covering graphs, but
these observations suggest that there may be an inherent weakness.
One can show that a covering graph
will inherit all the eigenvalues of a base graph.  Thus, imprimitivity
or a small spectral gap in a base graph yield the same
properties in a covering graph.  
This argument reverses the usual concern with graph covers: that
pseudo-codewords---essentially codewords from a covering graph of $B$---can
influence the performance of the SPA on $B$.  When check nodes have degree
$2$, the only pseudo-codewords are constant vectors, and the only
extremal pseudo-codewords are  the all-$0$ vector and the all-$1$
vector. In this case, analysis of pseudo-codewords can say nothing
about variation in decoding performance.  Yet there are clear
differences due to imprimitivity, and to other, more subtle, properties,
as the following examples illustrate.

\subsubsection*{Some Experiments}
The bipartite graph derived from the following graph is a $2$-fold cover
of $K_{3,2}$.
\vspace*{.4in}
\[
\xymatrix{
\bullet \ar@/^1pc/@{<->}[r] \ar@/_1pc/@{<->}[r] \ar@/^3pc/@{<->}[rrr]
& \bullet \ar@{<->}[r] 
& \bullet \ar@/^1pc/@{<->}[r] \ar@/_1pc/@{<->}[r]
&\bullet
}
\]

\vspace*{.2in}
Figure~\ref{f:undirected_bipartite} shows the bipartite graph derived from
$K_4$, which is not a cover of $K_{32}$.  The flow matrix for the
$2$-fold cover is imprimitive of index $2$, while the flow matrix for
the graph derived from $K_4$  is primitive.  
Figure~\ref{f:length4performance} shows
performance of the SPA with convergence parameters $\eps=10^{-2}$,
$10^{-4}$ and $10^{-8}$.  The curve for the bipartite graph from $K_4$
maintains roughly a $0.5$ dB gain over the $2$-fold cover at all
values of $\eps$. 
Figure~\ref{f:length10performance}
shows performance curves,using the same values for $\eps$,
for a five-fold cover of $K_{32}$, which is imprimitive of index $2$, and the bipartite graph derived from
the Peterson graph, which has a primitive
flow matrix.  We see an improvement of roughly $0.5$ dB at $\eps= 10^{-8}$,
but with less discriminating choices of convergence parameter the Peterson graph actually performs  worse.

\begin{figure}[htb]
\centering
\epsfig{figure=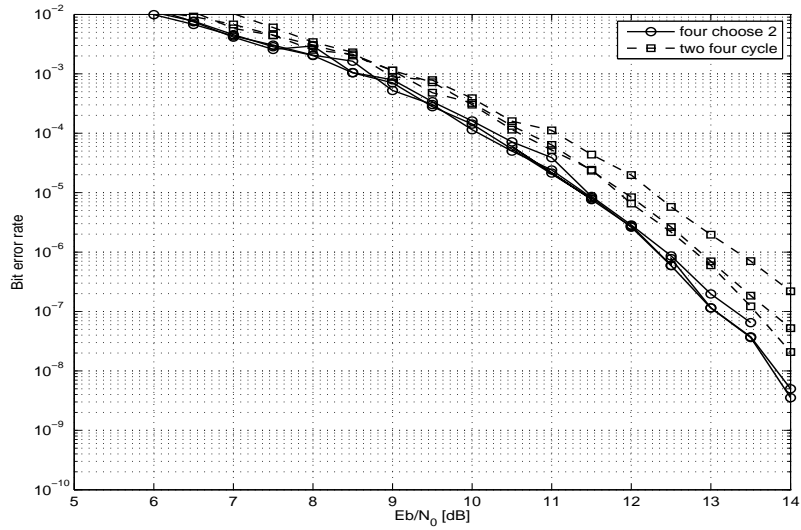,height=2.8in, width=4.2in, angle=0}
\caption{\label{f:length4performance}
The performance of a 2-cover of $K_{32} $ and the bipartite graph derived from $K_4$ 
for $\eps=10^{-2}$, $10^{-4}$, $10^{-8}$.}
\end{figure}

\begin{figure}[hbt]
\centering
\epsfig{figure=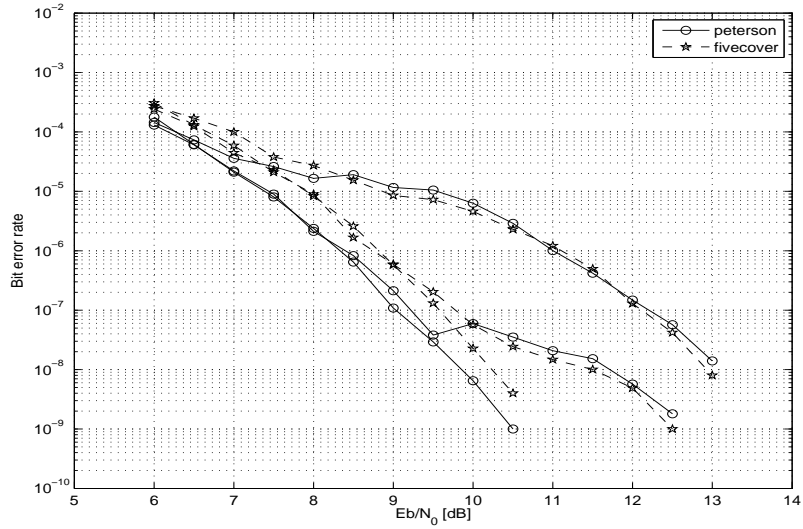,height=2.8in, width=4.2in, angle=0}
\caption{\label{f:length10performance}
The performance of a 5-cover of $K_{32} $ and the bipartite graph derived from the Peterson graph for $\eps=10^{-2}$, $10^{-4}$, $10^{-8}$.}
\end{figure}


Both $K_4$ and the Peterson graph are $3$-{\em cages}, that is, $3$-regular graphs having the minimum number of nodes for their respective 
girths, which are $3$ and $5$.
(Note that the bipartite graphs derived from these graphs have girth
that is twice as large.)
One might attribute the performance gain to the large girth, but 
by itself this is not an assurance of asymptotic optimality.
Below are the connected three-to-one covers of $K_{3,2}$. One has
girth~$4$; one has two $2$-cycles; and one has three $2$-cycles.  
(Again, double these values for the bipartite graph derived from these graphs.)

\vspace*{.6in}
\[
\xymatrix{
\bullet \ar@{<->}[r] \ar@/_2pc/@{<->}[rrr] \ar@/^4pc/@{<->}[rrrrr]
& \bullet \ar@{<->}[r]  \ar@/^2pc/@{<->}[rrr]
& \bullet \ar@{<->}[r] \ar@/_2pc/@{<->}[rrr]
&\bullet \ar@{<->}[r] 
& \bullet \ar@{<->}[r] 
&\bullet
}\]

\vspace{.6in}

\[
\xymatrix{
\bullet \ar@/^1pc/@{<->}[r] \ar@/_1pc/@{<->}[r] \ar@/_2pc/@{<->}[rrr]
& \bullet \ar@{<->}[r] 
& \bullet \ar@{<->}[r] \ar@/^2pc/@{<->}[rrr]
&\bullet \ar@{<->}[r] 
& \bullet \ar@/^1pc/@{<->}[r] \ar@/_1pc/@{<->}[r]
&\bullet
}
\]
\vspace{.6in}

\[\xymatrix{
\bullet \ar@/^1pc/@{<->}[r] \ar@/_1pc/@{<->}[r] \ar@/^4pc/@{<->}[rrrrr]
& \bullet \ar@{<->}[r] 
& \bullet \ar@/^1pc/@{<->}[r] \ar@/_1pc/@{<->}[r]
&\bullet \ar@{<->}[r] 
& \bullet \ar@/^1pc/@{<->}[r] \ar@/_1pc/@{<->}[r]
&\bullet
}\]

\vspace*{.2in}
The girth-~$4$ graph is actually a  $(3,4)$-cage.
Since all these graphs are covers of $K_{3,2}$ they are all
imprimitive and have the same asymptotic performance, as can be seen in 
Figure~\ref{f:length6performance}, where the performace
curves overlay each other  for $\eps= 10^{-8}$.  
Observe that for  $\eps= 10^{-2}$ and $10^{-4}$, the girth-$4$
graph is substantially better.  
We have no exact explanation for this, but we suspect it is
due to some property of the second largest eigenvalues.
For each of these graphs, there are $8$
eigenvalues (counted with multiplicity) of modulus $\sqrt{2}$, 
but the number of {\em distinct} eigenvalues of modulus $\sqrt{2}$ differs: 
the girth-$4$ graph has  only two, namely $\pm \sqrt{2}i$, while
the graph with two $2$-cycles has $6$ such and the graph with three $2$-cycles
has~4.   

\begin{figure}[htb]

\centering
\epsfig{figure=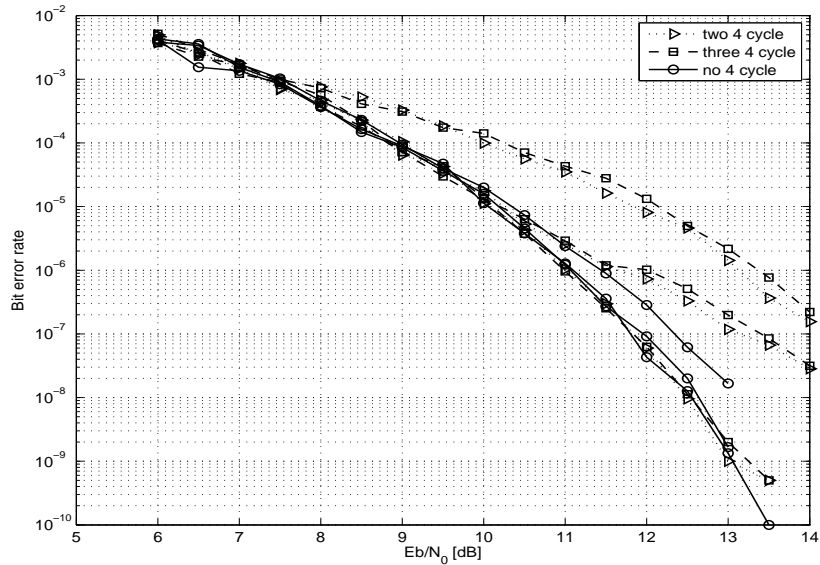,height=3in, width=4.3in, angle=0}
\caption{\label{f:length6performance}
The performance of three 3-covers of $K_{32} $ for 
$\eps=10^{-2}$, $10^{-4}$, $10^{-8}$.}
\end{figure}

The preceeding figures also show that the numerical precision used for
implementing the SPA  plays a strong role in the performance.  
In particular, at $\eps = 10^{-4}$ there is a very identifiable error
floor in several of  the performance curves. 
This corroborates a phenomenon that we observed in our experimental results on $(3,6)$-regular codes
of lengths $282$ and $1002$ conducted for \cite{OSBreWol05}.  
In the course of determining the performance of a number of codes, we collected all
vectors that caused decoding failure.  Subsequent testing
revealed that the vast majority of these vectors could be 
successfully decoded by reducing the termination parameter $\eps$.  
This raises two important questions:  To what  extent is the error floor
an artifact of numerical precision as opposed to non-convergence of
the algorithm (as is the case in the above examples)?  
Is there a method for creating graphs that are relatively
resistant to numerical imprecision?

\section*{Acknowledgements} Joshua Lampkins contributed to the research for this article.  In particular, he noticed 
the impact on the convergence of the sum-product algorithm when the graph $G$ is bipartite.

\bibliographystyle{plain}
\bibliography{deg2}
\end{document}